\documentclass[journal]{IEEEtran}

\usepackage{graphicx}      % include this line if your document contains figures
\usepackage{latexsym}
\usepackage{amsmath}
\usepackage{amsfonts}
\usepackage{epstopdf}
\usepackage{amssymb}

\usepackage{amsthm}

\newtheorem{definition}{Definition}

\newtheorem{theorem}{Theorem}%[section]
\newtheorem{lemma}[theorem]{Lemma}
\newtheorem{proposition}[theorem]{Proposition}

\newtheorem{remark}[theorem]{Remark }

\usepackage{color}

\usepackage{xcolor}
\definecolor{verde}{rgb}{0.2,0.6,0.6}

\usepackage[colorlinks=true,linkcolor=blue,urlcolor=blue,citecolor=blue]{hyperref}

\newcommand{\be}{\begin{equation}}
\newcommand{\ee}{\end{equation}}

\newcommand{\bea}{\begin{eqnarray}}
\newcommand{\eea}{\end{eqnarray}}
\newcommand{\bean}{\begin{eqnarray*}}
\newcommand{\eean}{\end{eqnarray*}}

\IEEEoverridecommandlockouts                              % This command is only needed if 
\begin{document}

\title{\LARGE \bf
On the detection and identification of edge  disconnections in a
multi-agent consensus network}

 \author{Gianfranco Parlangeli and Maria Elena Valcher 
 \thanks{G. Parlangeli is with the Dipartimento di Ingegneria dell'Innovazione, Universit\`a del Salento,
    Via per Monteroni, 73100 Lecce, Italy,
     e-mail: \texttt{gianfranco.parlangeli@unisalento.it}. M.E. Valcher is with
 the Dipartimento di Ingegneria dell'Informazione
 Universit\`a di Padova, 
    via Gradenigo 6B, 35131 Padova, Italy, e-mail:  \texttt{meme@dei.unipd.it}.}
   } 
 \maketitle

\thispagestyle{empty}
\pagestyle{empty}

%%%%%%%%%%%%%%%%%%%%%%%%%%%%%%%%%%%%%%%%%%%%%%%%%%%%%%%%%%%%%%%%%%%%%%%%%%%%%%%%%
 \begin{abstract} In this paper we   investigate the problem of
the sudden disconnection of an edge in a discrete-time multi-agent consensus network.
{\color{black} If the graph remains strongly connected, the multi-agent
  system still achieves consensus, but} in general,  unless   the information exchange between each pair of agents is symmetric, the agents' states converge to a drifted
value of the original consensus value. 
% We first show that if the communication network is symmetric, an edge disconnection that does not affect the connectedness of the whole network does not even affect the final consensus value. On the contrary, if the communication network is not symmetric, then, in general, an edge disconnection that preserves the strong connectedness of the communication graph nonetheless may lead to a consensus value that is different from the original one.
{\color{black} Consequently the edge disconnection can go unnoticed. 
In this paper the problems of 
detecting an edge disconnection and of identifying in a finite number of steps the exact edge that got disconnected are investigated. Necessary and sufficient conditions for both problems to be solvable are presented, both in case all the agents' states are available and in case only a subset of the agents' states
is measured.  Finally, an example of a network of 7
  agents  is provided, to illustrate some of the theoretical results
derived in the paper.}
\end{abstract}


%\begin{keywords} Multi-agent system, consensus, strongly connected network, Laplacian, fault detection and identification.
%\end{keywords}

%
%
%%%%%%%%%%%%%%%%%%%%%%%%%%%%%%%%%%%%%%%%%%%%%%%%%%%%%%%%%%%%%%%%%%%%%%%%%%%%%%%%%
\section{Introduction}
 The technology advances   %occurred in 
of the last decades
  brought brand new opportunities in several engineering areas and
  stimulated a significant thrust of research  in
  communications, information processing and control theory %to catch them
  \cite{antsaklis2007special,gharavi2003special,zhao2004wireless},
The availability of miniaturized low-cost processing and
integrated communication devices allowing peer-to-peer communication pushed a renewed interest in the
analysis and design of distributed heterogeneous sensor networks and in the
cooperative control of networks of autonomous agents
\cite{baillieul2007control,chong2003sensor}. 

For a network of agents,
\emph{consensus} is a fundamental target, useful for coordination,
%among agents, 
and a key tool 
to achieve decentralized
architectures
\cite{OlfatiFaxMurray,SurveyConsensus2005, RenBeardAtkins}.
\emph{Consensus} refers to the situation when   agents
achieve a common decision
 on a local variable, which is updated performing local
computation and exchanging local information
\cite{OlfatiFaxMurray}.
Under the assumption that each agent follows a prescribed protocol,
the use of  such local variable makes %it possible to let 
the system behave as if the agents
were fully connected (i.e. as if the communication graph were complete)
\cite{OF-Murray2004}, and this condition enables a variety of
applications in a wide range of fields \cite{DistCtrlRobotNetw},
\cite{leonard2007collective}, \cite{ren2007distributed}, 
\cite{ren2010distributed}, %\cite{santini2017consensus}, 
\cite{wah2007synchronization}. % {\color{black} Consensus-based applications range in  a wide number of technological
% fields, such as electrical power grids \cite{dg15} and transportation networks
% \cite{santini2017consensus}, cooperative robotics, surveillance, and environmental monitoring
% \cite{DistCtrlRobotNetw}. }
However, the adherence of each agent to the agreed protocol may fail for a
number of reasons, and several research directions have been explored to achieve  {\color{black} robust
  consensus in the presence of intermittent transmissions, transmission errors, faults or noise \cite{li2017robust}, \cite{miah2014nonuniform}, \cite{qin2017recent,you2013consensus}. 

Malfunctions in a network,  may be temporary or intermittent,  
%such as  computational errors, loss or corruption of a transmitted
%message, temporary adversarial environmental conditions, 
casual or intentional (for instance, they may  be the results of a cyber-attack \cite{pasqualetti2015control}),
and a long stream of  research has been devoted to 
address fault-detection and identification (FDI) problems or fault-tolerant control strategies for multi-agent systems, see \cite[Section~II-B]{qin2017recent}
for an extensive literature review. 
FDI algorithms for multi-agent systems split into two classes:
centralized  and distributed ones. In the first case \cite{costanzo2017using,pasqualettiTAC2013,rahimian2015failure},
a central unit gathers all the   system information to perform the fault diagnosis
algorithm. In the
  distributed case \cite{davoodi2014distributed,pasqualettiTAC2013,teixeira2014distributed}, all agents run some local fault
detection algorithm, based on  local information and on
the signals received from   neighbouring agents, and they coordinate to decide whether and where the fault occurred.

Most of the literature on  FDI for multi-agent systems, however,
assumes that faults act additively either on the state-update
equations or   on  the signals exchanged by pairs of agents
\cite{davoodi2014distributed,pasqualetticarliIFAC,pasqualettiTAC2013,teixeira2014distributed}. Even
if an edge   disconnection can be modelled in this way, this set-up is
too general and it does not exploit the correlation between the fault
signal and the state evolution that characterizes faults resulting
from edge disconnections. As a result, FDI conditions deduced in the
general set-up do not exploit the special nature of these faults, thus
leading to conservative conditions for a successful  
  and prompt 
diagnosis when dealing with   edge disconnection. 

On the other hand, in general, the problem of detecting an edge disconnection has been investigated for multi-agent systems that are not necessarily consensus networks. For instance,
in \cite{BattistelliTesi} and \cite{patil2019indiscernible} the
possibility of detecting an edge or a node disconnection in a
multi-agent system is investigated, and the concepts of  discernibility from the states or the outputs are investigated.
In  \cite{Xue} the problem of detecting an edge disconnection is addressed for  a diffusive network, by resorting to an impulsive input 
applied at one specific node. Simple graph conditions allow to determine whether the problem is solvable. The residual signal is generated at the specific node by assuming that the whole network topology is known, and it depends on the whole state trajectory. 
The authors of  \cite{costanzo2017using}      use stochastic
techniques and propose an algorithm based on the full state knowledge
to promptly detect abrupt topological
changes of the network such as a link failure, creation, or degradation.
In \cite{dhal2013link,dhal2015detecting,torres2015detecting} 
the detection of a  link disconnection in a network  from noisy measurements at a single node is investigated. By making use of  a Maximum A Posteriori Probability technique, conditions for  asymptotic detection, in terms of the network spectrum and graph, are derived. If the detector does not know the whole state of the system, perfect detection
is not possible. Note, also, that none of the previous references explicitly addresses the identification problem.

Few contributions have specifically addressed the problem of detecting and/or identifying 
a link disconnection in a consensus network, by making use of the network properties and of the very specific change of the network description that results from this fault.  Specifically,
in \cite{RahimianCDC2012}  a multi-agent
system subjected to multiple communication link failures is considered,
with the goal of deriving useful design guidelines for reliable and
fault-tolerant multi-agent networks. The  possibility of detecting multiple link failures corresponding to {\em at least some} initial conditions is characterised in terms of the communication graph properties.
In  \cite{pandey2019diffusion} 
an algorithm for diffusion of information among nodes reaching a weighted
agreement is proposed, which is informative of the topology of the network and can
be  applied to the detection  (but in general not the identification) of  a link failure.
In \cite{rahimian2015failure} a method to detect and identify link
failures in a network of continuous-time homogeneous agents, 
with a weighted and directed communication graph, is proposed, assuming that  only
the output responses of a subset of nodes are available. Jump discontinuities in
the output derivatives  are used to detect  link failures. The order of the derivative at which
the discontinuity  is observed depends on the relative
degree of the agents transfer matrix, and on  the distance of the observation point
from the disconnected edge. A graph based algorithm to identify the link is presented.
 It is worth remarking that all the previous references adopt a centralised approach to the problem solution.

This paper    focuses on the problem of detecting and identifying
an edge removal   within a network of agents reaching
consensus.  
This type of fault may result in a disconnected network, for which consensus is no longer achievable,
and in this case the fault cannot go unnoticed but it can be detected significantly
later than it occurred. Alternatively, the communication network may remain connected, but a different algorithm 
is performed with respect to the original one and this leads to agree on a final value which is different
%a different weighted average of the initial conditions with respect to 
from the one the original network would have achieved.
For this reason it is fundamental to detect and possibly identify disconnected edges, so that they can be promptly restored.
We assume, as it was done in the large majority of the aforementioned 
references, that the agents represent devices with limited functionalities. In general they have not the
capability of detecting an edge disconnection, not even when they are directly affected by it (by this meaning that they are either the transmitter or the receiver at the extremes of the disconnected edge), and hence they cannot send an alarm signal. This assumption is extremely realistic when dealing with large networks of cheap sensors, for instance, whose software is configured to perform very basic tasks. Accordingly, we assume that 
the detection and identification of a broken edge cannot be performed
  at local level   by means of a distributed algorithm that involves a subset of the agents
    and we propose   centralised algorithms   that are able to identify a link disconnection from the knowledge of either the entire state evolution  or
the state evolution of a subset of its agents.
In fact, a distributed FDI is typically possible
only under the condition 
that a selected group of agents is able to both generate a residual signal and to communicate   with each other  to reach a decision about the occurrence of a fault. This requires additional features with respect to the basic ones that are imposed by the consensus algorithm.
%Nonetheless the malfunctioning of any such device, or an altered transmission within the network, affects the final outcome of the distributed algorithm and hence needs to be promptly identified and fixed. 
%In fact, in some of the previous references a distributed FDI was possible
%only under the condition 
% that the selected agents that were able to generate some residual signal were also able to communicate one with the other in order to reach a decision about the occurrence of a fault. On the other hand, in most of the cases a distributed FDI had been obtained by resorting to a full-state observer for each of the selected agents
In addition to the fact that for certain multi-agent systems, the physical nature of the devices
does not offer alternatives, the choice of adopting a centralised approach has several motivations: first of all, what can be obtained in a centralised way 
always represents a benchmark that distributed solutions try to
approach as much as they can, and since a clear analysis of this
problem is not available yet in the literature, we believe that this
is the first goal to achieve. Secondly, a centralised solution
requires weaker properties in terms of
observability/reconstructibility  than the ones that a distributed solution would impose on the single agents.
% **involved in the distributed check**. {\color{red} non ho capito: vuoi togliere la frase? non vorrei che sembrasse che noi mettiamo in opposizione la soluzione centralizzata e quella in cui un solo agente fa tutto il lavoro, senno' diciamo un'ovvieta' - Anche una frase di prima e' stata modificata cosi' da rendere questo aspetto piu' ambiguo} {\color{verde} Era solo per alleggerire la frase, pensavo si potesse
%  sottintendere perche' prima c'e' 'distributed solution'. } 
  It is often the case that a wise choice of a small set of agents whose states need to be monitored in a centralised way allows an effective detection and identification, but those same agents would not be able to independently perform FDI. This aspect will be better clarified  at the end of the paper.
Finally, 
%not all network structures are compatible with a  distributed algorithm, where  selected agents estimate only part of the overall state vector and generate a residual signal based on it. S
several distributed FDI algorithms rely on the estimate of the overall
state vector, and
this is extremely demanding from a computational point of view, as well as not realistic, since it
presumes that single agents may have a complete knowledge of the
overall system structure. Solutions that require the selected agents
to estimate only portions of the state vector typically impose very
strong conditions of the structural properties of the overall system
  and they are   strongly dependent on the specific system
    structure. 
%{\color{red} in  caso citare
%un paper di Bullo, Dorfler, Pasqualetti} {\color{verde} le mie
%considerazioni vengono da 'Solving a system of linear equations: From
%centralized to distributed algorithms' di Weii Ren et al, ma anche
%pensando al paper di pasqualetti..Ad es., sistemi singolari.}} 
}
% Reference \cite{kleinberg2008network}   focuses on graph-theoretical
%conditions to select a {\color{black} subset of nodes that allow} to detect 
%$\kappa$ simultaneous failures {\color{black} affecting either} edges or nodes.

The paper structure is the following one. Section II presents the
class of discrete-time multi-agent systems and the general setting of the problem. In Section III the effects of the disconnection of a single link are explored.
The problem of detecting an edge disconnection is tackled in Section
IV, and it is split into
four subsections. The discernibility {\color{black}   of two networks,  one of them obtained from the other 
as a result of the disconnection
of a link,  is studied in Subsection IV.A from a theoretical
standpoint, assuming that  the whole state vector is
available for measurement. Some special situations that prevent   discernibility are discussed}. In  Subsection IV.B  an
algorithm for link failure detection and isolation is introduced and
{\color{black} it is shown that, if discernibility conditions are
satisfied,     the algorithm can} detect each fault and isolate it under an
additional condition. {\color{black} Section V is devoted to the discernibility
problem when only a subset of
nodes is available and in Subsection V.A  an algorithm for  detecting the edge disconnection under these conditions
is provided.  Also in this case, if discernibility conditions are
satisfied, the algorithm} is able to detect each fault, while fault isolation is ensured under an
additional condition. {\color{black} In Section VI  the case of a network of 7 nodes with only three states available for measurement is considered,  and the results derived in the paper are illustrated for different edge disconnections.}

{\color{black} 
 The results provided in this paper 
 have been inspired by \cite{BattistelliTesi}, where the concept of discernibility of a multi-agent system
 from the (faulty) one resulting from an edge or a node disconnection have been first investigated.
 Compared to \cite{BattistelliTesi}, we tailor our analysis to a consensus network and hence account for the fact
 that discernibility in a consensus network does not reduce to the observability of a special matrix pair associated with the healthy and the faulty systems.  Indeed, the dominant eigenvalue/eigenvector that ensure consensus do not change after the edge disconnection and need to be separately accounted for.
 In addition we have explicitly addressed the identification problem and proposed an algorithm to identify which specific link of the system got disconnected. To achieve this goal we have proposed a residual generator that is based on a (full-order) dead-beat observer. This solution has the great advantage of zeroing the effects of the initial conditions in a finite number of steps, thus making it possible to promptly detect an edge disconnection in the minimal   number of steps, even when the effects of such a disconnection are small and hence may be erroneously interpreted as the effect of disturbances.
 This manuscript
 extends} the edge failure analysis
first performed in \cite{ECC2019} and then summarized in \cite{ICSTCC2019} {\color{black} (this 
latter paper  presents  some preliminary results about node disconnection and compares such results with those obtained in \cite{ECC2019} for the edge disconnection)}  in the following way. 
First of all in this paper we address the case of a directed communication graph, rather than an undirected one. 
{\color{black}  While Sections II and III represent slightly modified versions of the original Sections II and III in \cite{ECC2019}, the analysis in the following sections is sigificantly improved and extended compared to the one presented in \cite{ECC2019}. The proof of Proposition 3 is new. The whole part of subsection IV.A, after  Remark 6 (in particular, Proposition 7), and subsection IV.B (in particular, Proposition 8) are original. Since this last part deals with the problem of  providing a method to detect and identify the specific edge that got disconnected, we believe that there is a significant added value with respect to the original conference paper \cite{ECC2019}. Also, the proof of Proposition 9 in Section V is original, and the whole Subsection 
V.A, providing conditions for fault detection and identification of an edge disconnection when only a subset of the states is available, is original. Finally, Section VI provides a useful illustrative example and is new.}
We believe that the current results provide a first meaningful step  toward the final goal of first determining necessary and sufficient conditions for detecting and identifying all kinds of {\color{black} edge/mode faults,  for  general classes of homogeneous multi-agent systems, and then  designing practical algorithms to perform their}   detection and identification.

\smallskip

\noindent {\bf Notation.}
 ${\mathbb Z}_+$ and $\mathbb{R}_+$ denote the set of nonnegative integer and real numbers, respectively.    
Given $k, n\in {\mathbb Z}_+, k <n,$ we denote by $[k,n]$ the set of integers $\{k, k+1, \dots, n\}$.
We let $\mathbf{e}_i$ denote
the $i$-th element of the canonical basis in ${\mathbb R}^k$ ($k$ being clear from the context), 
with all entries equal to zero except for   the $i$-th one which is unitary.
 ${\bf 1}_k$ and ${\bf 0}_k$ 
 denote the $k$-dimensional real vector  
 whose entries are all $1$ or all $0$, respectively.
Given a  real matrix $A$,  the $(i,j)$-th entry of $A$ 
%(corresponding to its $i$-th row and $j$-th column) 
is denoted either by
$a_{ij}$ or by $[A]_{ij}$,      and its {\em transpose} by $A^\top$.
Given a vector ${\bf v}$, the $i$-th entry of ${\bf v}$ is denoted by $v_i$ or by $[{\bf v}]_i$.
%Given $A\in\mathbb{R}^{n\times n}$, $p_A(z):=\det{(zI_n - A)}$ denotes the {\em characteristic
%  polynomial} of $A$.     
   The \emph{spectrum} of
  $A\in {\mathbb R}^{n \times n}$, denoted by $\sigma(A)$, is the set of its eigenvalues and the \emph{spectral radius} of
  $A$, denoted by $\rho_A$, is the maximum modulus of the elements of $\sigma(A)$.  
For a {\em nonnegative matrix} $A\in {\mathbb R}_+^{n\times n}$, i.e., a
  matrix whose entries are nonnegative real numbers, the spectral
  radius is always an eigenvalue. 
  A nonnegative and nonzero matrix is called {\em positive}, while a matrix whose entries are all
    positive is called  {\em strictly positive}. {\em Nonnegative, positive and strictly positive vectors} are analogously defined. A positive matrix
 $A \in \mathbb{R}_+^{n\times n}, n>1,$ is {\em irreducible} if no permutation matrix $P$ can be found such that
 $$P^\top A P= \begin{bmatrix}A_{11} & A_{12}\cr 0 & A_{22}\end{bmatrix},$$
 where $A_{11}$ and $A_{22}$ are square {\color{black} (non-vacuous)} matrices.
  By the Perron-Frobenius theorem \cite{Berman-Plemmons,BookFarina,OlfatiFaxMurray}, for an irreducible positive matrix
 $A$ the spectral radius $\rho_A$ is a simple real %(strictly) 
 dominant eigenvalue,
 and the corresponding left and right eigenvectors are {\em strictly positive}. Positive eigenvectors of a positive irreducible matrix necessarily correspond to the spectral radius.
%  Given two positive matrices $A_1,
%  A_2\in {\mathbb R}^{n\times n}$ the inequality $A_1>   A_2$ means that each entry of the matrix on the left hand-side is greater than or equal to the corresponding entry of 
%  the matrix on the right hand-side, and the inequality is strict in
%  at least one case.
%   A polynomial $a(z)\in {\mathbb R}[z]$ is said to be {\em Schur} if
%  $a(\lambda)=0$ for some $\lambda \in {\mathbb C}$ implies
%  $|\lambda| < 1$. 
%Given   $\alpha_1, \dots, \alpha_n$, the symbol ${\rm diag}\{\alpha_1, \dots, \alpha_n\}$ denotes the $n$-dimensional diagonal matrix
%whose $(i,i)$-th entry is equal to $\alpha_i$. Also, if ${\bf v}$ is an $n$-dimensional vector with $i$th entry $v_i$, then 
%${\rm diag}({\bf v}) := {\rm diag}\{v_1, v_2, \dots, v_n\}$.
%Given  %a vector 
%${\bf v} =[v_i]\in {\mathbb R}^n$,   the symbol ${\rm diag}({\bf v})$ denotes the $n$-dimensional diagonal matrix
%whose $(i,i)$-th entry is %equal to 
%$v_i$.  
% {\color{verde} potremmo non definire l'operazione 'diag(v)',
%  la usiamo una volta per la def della matrice laplaciana}

A {\em directed weighted graph}   ${\mathcal G}$ is a triple $({\mathcal V}, {\mathcal E}, {\color{black}{\mathcal W}})$, where ${\mathcal V}=[1,N]$ is the set of vertices, ${\mathcal E}\subseteq {\mathcal V} \times {\mathcal V}$ is the set of arcs, and ${\color{black}{\mathcal W}}$ is the matrix of the  weights of ${\mathcal G}$. 
${\color{black}{\mathcal W}}$ is called {\em adjacency matrix} of the graph. The $(i,j)$-th entry of ${\color{black}{\mathcal W}}$, $[{\color{black}{\mathcal W}}]_{ij}$, is nonzero if and only if the   arc $(j,i)$ belongs to ${\mathcal E}$.   
We assume that $[ {\color{black}{\mathcal W}}]_{ii} =0,$ for all $i\in [1,N]$,
namely there are no self-loops.  A {\em weighted graph}   ${\mathcal G}=({\mathcal V}, {\mathcal E}, {\color{black}{\mathcal W}})$ is {\em undirected } if the
arcs are bidirectional, namely  $(j,i)\in {\mathcal E}$ if and only if $(i,j)\in
{\mathcal E}$ and  $[ {\color{black}{\mathcal W}}]_{ij} =[ {\color{black}{\mathcal W}}]_{ji}.$ Therefore
 for undirected graphs    ${\color{black}{\mathcal W}}={\color{black}{\mathcal W}}^T$.

 A {\em path connecting $j$ and $i$} is an ordered
sequence of arcs $(j,i_1), (i_1,i_2), \dots, (i_{k-1},i_k),(i_k,i)\in
{\mathcal E}$. % ; a  {\em cycle} is a closed path, namely a path where $j$ and
% $i$ coincide
  A directed (resp. undirected) graph ${\mathcal G}$ is {\em
    strongly connected} ({\em connected}) if, for every
  pair of vertices $j$ and $i$, there is a path connecting them. 
${\mathcal G}$ is {\color{black} strongly connected  (connected)}   if and only if its {\color{black} (symmetric)} adjacency matrix ${\color{black}{\mathcal W}}$ is irreducible.
\newline The {\em Laplacian matrix}   \cite{Fiedler_consensus} associated with the adjacency matrix ${\color{black}{\mathcal W}}$ %, the  matrix of the signed weights of ${\mathcal G}$, 
is
defined as
${\mathcal L} := {\mathcal C} - {\color{black}{\mathcal W}},
$
where ${\mathcal C}$ is the (diagonal)  connectivity matrix,
  whose diagonal entries are the sums  of the corresponding row entries of
 ${\color{black}{\mathcal W}}$,
namely 
%$[{\mathcal C}]_{ii} = \sum_{(j,i)\in {\mathcal E}} [{\color{black}{\mathcal W}}]_{ij},   \forall i\in [1,N],$
%or, equivalently, 
%${\mathcal C}={\rm diag}({\color{black}{\mathcal W}}{\bf 1}_N)$
{\color{black} $[{\mathcal C}]_{ii} = \sum_{j=1}^N [{\mathcal W}]_{ij},
  \forall i\in [1,N]$.}
Clearly, by the way the Laplacian has been defined ${\mathcal L} {\bf
  1}_N = {\bf 0}_N$.  Also, ${\mathcal L}$ is irreducible if and only if ${\color{black}{\mathcal W}}$ is irreducible.  If ${\mathcal G}$ is undirected
  then the associated Laplacian ${\mathcal L}$ is  symmetric.
 \\
% The graph ${\mathcal G}$ is {\em strongly connected} if,
%          for every pair of vertices $j$ and $i$, there is {\color{black} a} path,
%          namely an ordered sequence of arcs $(j,i_1), (i_1,i_2),
%          \dots, (i_{k-1},i_k),(i_k,i)\in {\mathcal E}$, connecting
%          them.
 A family   $\pi=\{ {\mathcal V}_1, .., {\mathcal V}_k\}$ of 
non-empty subsets of  ${\mathcal V}$ such that $\cup_{i=1}^{k}{\mathcal
  V}_i={\mathcal V}$ and ${\mathcal V}_i \cap {\mathcal V}_j=\emptyset$, $\forall
i\neq j$, is called  a {\em partition} of the vertex set ${\mathcal V}$. 
When so,  
${\mathcal V}_i$ is called the $i$-th \emph{cell} of the partition $\pi$, and
%for every cell ${\mathcal V}_i$, 
the vector ${\mathbf
  x}\in \mathbb{R}^N$
satisfying {\color{black}$x_\ell=1$ if $\ell\in {\mathcal V}_i$ and $x_\ell=0$ if $\ell\notin {\mathcal V}_i$,  i.e., ${\bf x}=\sum_{\ell\in {\mathcal V}_i} {\bf e}_\ell$,}
 is called the \emph{characteristic vector} of the $i$-th cell ${\mathcal V}_i$. Finally, %given the partition $\pi$, 
 the
 matrix $P_{\pi}\in\mathbb{R}^{N\times k}$,
 whose $i$-th column is the characteristic vector of the subset
 ${\mathcal V}_i$ %constituting a partition $\pi$ 
is called the \emph{characteristic matrix} of $\pi$. 
%An {\em undirected weighted graph}   ${\mathcal G}=({\mathcal V}, {\mathcal E}, {\color{black}{\mathcal W}})$ is a weighted graph for which 
%an     arc $(j,i)$ belongs to ${\mathcal E}$ if and only if the arc $(i,j)$ belongs to ${\mathcal E}$ and if so the arcs have the same weight.
%This amounts to saying that the adjacency matrix ${\color{black}{\mathcal W}}$, and hence the associated Laplacian ${\mathcal L}$, are  symmetric.

%\section {The multi-agent consensus network}
%\label{ProblemSetup}
%
\section {Problem setup}
\label{ProblemSetup}

Consider a multi-agent system consisting 
 of {\color{black}$N>2$} agents,  each of them  indexed in the integer
set $[1,N]$. The state of the $i$-th agent is described by the scalar  variable $x_i$ 
  that updates according to the following discrete-time linear state-space model
%\cite{OlfatiSaber07,
\cite{OF-Murray2004}:
$$x_i(t+1)=x_i(t)+v_i(t), \qquad t\in \mathbb{Z}_+,$$
where $v_i$ is %represents 
the input   of %that 
the $i$-th agent. % receives from the other agents, in order to update its state variable.
 The communication  among the $N$ agents is described
    by a     fixed directed graph ${\mathcal G}$ with 
adjacency matrix $ {\color{black}{\mathcal W}} %= {\color{black}{\mathcal W}}^\top 
\in {\mathbb R}^{N \times N}$.
The $(i,j)$-th entry, {\color{black} $i\ne j$,} of ${\color{black}{\mathcal W}}$ is positive, i.e.,   $[ {\color{black}{\mathcal W}}]_{ij}>0$, if there is information flowing from
agent $j$ to agent $i$, and  $[ {\color{black}{\mathcal W}}]_{ij}=0 $ otherwise. 
%We assume   $[ {\color{black}{\mathcal W}}]_{ii}=0,$  $\forall i \in \{1, \dots, N\}.$ 
% The time-invariance assumption on the communication topology 
%    is a restrictive hypothesis, but nonetheless a quite realistic one when  dealing 
%  with consensus problems for networks whose nodes have  fixed positions. 
%For instance, under this   assumption consensus algorithms have
%been effectively applied to provide distributed solutions to  
%fundamental issues for an electric grid, such as state estimation,
%economic dispatch and optimal power flow \cite{kar2014distributed}, thus providing important tools for its evolution
%to a smart grid. A related recent application field is heat and energy distribution
%in building automation systems \cite{kim2015consensus}.
%and we denote by
%$\mathcal{N}_i$ the set of neighbours of the $i$-th agent, namely the
%set of agents whose information is used by the $i$-th agent to update
%its state.
Each agent adopts the    (nearest
  neighbor linear) consensus protocol
\cite{OF-Murray2004}, %, \cite{ren2005coordination} 
which amounts to saying that the input %signal
 $v_i$ takes the   form:
\be
v_i(t)= \kappa \sum_{j=1}^N
 [ {\color{black}{\mathcal W}}]_{ij}(x_j(t)-x_i(t)),
\label{cons_protocol}
\ee
where $\kappa >0$ is a given real parameter known as {\em coupling
  strength}\footnote{{\color{black} If we regard the discrete-time system describing the $i$-th agent dynamics as the discretized version of the continuous-time equation $\dot x_i(t) =v_i(t)$,
    $\kappa$ represents the sampling time.}}. 
If we stack  the states of the agents in a 
single %$N$-dimensional 
state vector   ${\bf x}\in\mathbb{R}^N$, the overall multi-agent system %with consensus protocol
becomes
\be\label{sana}
{\bf x}(t+1)=(I_N -\kappa {\mathcal L}) {\bf x}(t) =: A {\bf x}(t),
\ee
%{\color{black} questa parte potrebbe essere accorciata e si potrebbe partire direttamente dalla precedente equazione, evitando di parlare di ${\mathcal L}$ }
where    
 ${\mathcal L} =[\ell_{ij}]  
% := {\rm diag}\{\sum_{j\ne1} [ {\color{black}{\mathcal W}}]_{1j},
% \dots, \sum_{j\ne N} [ {\color{black}{\mathcal W}}]_{Nj}\} -  {\color{black}{\mathcal W}}
  \in
 {\mathbb R}^{N \times N}$ is the %symmetric 
Laplacian
 associated with the adjacency matrix $
 {\color{black}{\mathcal W}}$. %(equivalently, with the communication graph). 
It is easy to deduce the relationship
     between eigenvalues $\lambda_A\in\sigma(A)$ and $\lambda_{\mathcal L}\in\sigma({\mathcal
       L})$, namely \be\label{eigAvsL} \lambda_A=1-\kappa
     \lambda_{\mathcal L}. \ee 
     In particular,  ${\mathcal L} {\bf
  1}_N = {\bf 0}_N$ implies that $A {\bf 1}_N = {\bf 1}_N$.

% Note that, 
% by the properties of the Laplacian, we have
%    $A {\bf 1}_N = {\bf 1}_N$, and hence $1$ is an eigenvalue of $A$. 
%   Also, it is immediate to see that, by the way $A$ and ${\mathcal L}$ are related, $A$ and ${\mathcal L}$ have all the eigenvectors
%    in common, {\color{black} and it is easy to deduce the relationship
%      between the two spectra $\sigma(A)$ and $\sigma({\mathcal
%        L})$. In particular, the unitary eigenvalue of $A$  corresponds
%      to the zero eigenvalue of ${\mathcal L}$}.

\noindent System \eqref{sana} can be used to describe a wide
variety of %practical applications as it describes the 
{\color{black}situations where} each agent/node   shares information with its neighbours with the final goal  of converging to a
common  {\color{black} decision.}
%constant value. 
If so, we refer to the multi-agent system as to a consensus network.
More formally,
system \eqref{sana} is a {\em consensus network} if for every   initial state ${\bf x}(0)$ there exists
$\alpha\in\mathbb{R}$ such that
  \be
 \lim_{t\to +\infty} {\bf x}(t)=  \alpha \mathbf{1}_N. \label{leaderless}
 \ee
 The constant
 $\alpha$ is called the \emph{consensus value} %(or {\em collective decision}) 
 \cite{OlfatiFaxMurray} for system \eqref{sana}, corresponding to the given initial state. 
If the agents' communication graph is
  strongly connected, namely the Laplacian ${\mathcal L}$ is irreducible \cite{OlfatiFaxMurray},    and  the coupling strength $\kappa$ satisfies
the following constraint:
\be
0 < \kappa < \frac{1}{\max_{i\in[1, N]} \ell_{ii}},
\label{constr_kappa}
\ee
$\ell_{ii}$ being the $i$-th diagonal entry of ${\mathcal L}$, 
system \eqref{sana}  is a consensus network  (see Theorem 2 in \cite{OlfatiFaxMurray}).
 Moreover, the consensus value
is   \be
\alpha =  {\bf w}_A^\top {\bf x}(0),
 \label{cons_value}
 \ee 
 where ${\bf w}_A$ is the left eigenvector of $A$ corresponding to $1$ and satisfying ${\bf w}_A^\top {\bf 1}_N=1$.
 In the special case when the graph is undirected and hence ${\mathcal L}$ and $A$ are  symmetric, ${\bf w}_A= \frac{1}{N}{\bf 1}_N$ and hence the consensus value
 is the average value of the  agents' initial conditions.
% This means
%Note that the final value on which the agents agree is a linear
%function of the initial conditions $x_i(0), i\in{\color{black}[1,N]}$, of
%the agents; %, and 
%this kind of agreement is known %in the literature 
%as {\em weighted-average consensus} \cite{OlfatiSaber07}.
\newline {\bf Assumption 1.}  {\em In the following we steadily assume that 
  ${\mathcal L}$ is irreducible and $\kappa$ satisfies the inequalities in \eqref{constr_kappa}.
Consequently, $A=I_N-\kappa {\mathcal L}$ is a positive irreducible matrix.  Perron-Frobenius theorem and condition $A {\bf 1}_N = {\bf
  1}_N$ ensure  
  that $1$ is a simple dominant eigenvalue of $A$. The eigenspace associated with the unitary eigenvalue is $\langle {\bf 1}_N\rangle$, and the positive eigenvectors of $A$ necessarily correspond to $\lambda=1$ and hence   belong to $\langle {\bf 1}_N\rangle$.}
\smallskip

 In this paper we investigate the effects of {\color{black} an edge  % or a  node  
disconnection on a consensus  network, and the possibility of detecting  and identifying such a failure.  }

\section{Consensus after an edge disconnection}
\label{sec3}

If the communication from  agent $r$ to agent $h$ is interrupted,
namely the arc $(r,h), r\ne h,$ is disconnected, then the Laplacian
$\bar {\mathcal L}$ of the new digraph\footnote{In
  order not to make the notation heavy, in this part of the paper we
  denote by $\bar {\mathcal G}$ the new digraph, by
    $\bar{\mathcal L}$ its Laplacian and by $\bar A$ the new system
  matrix, without highlighting in the notation the specific link that
  gets disconnected. Later on, we will modify the notation to distinguish the effects of different edge disconnections.}  $\bar {\mathcal G}$
  is related to the Laplacian ${\mathcal L}=[\ell_{ij}]$ of the original digraph ${\mathcal G}$ by the relationship
$$\bar {\mathcal L} = {\mathcal L} + \ell_{hr} {\bf e}_h{\bf e}_h^\top - \ell_{hr} {\bf e}_h{\bf e}_r^\top = {\mathcal L} + \ell_{hr} {\bf e}_h [{\bf e}_h -{\bf e}_r]^\top,$$
where $\ell_{hr} = - [{\color{black}{\mathcal W}}]_{hr} < 0$.
Consequently,  the new  state-update matrix of the multi-agent system becomes
\be 
\bar A := I_N - \kappa \bar {\mathcal L} = %I_N - \kappa {\mathcal L}  - \kappa \ell_{hr} {\bf e}_h({\bf e}_h -{\bf e}_r)^\top \nonumber \\
%&=& 
A - \kappa \ell_{hr} {\bf e}_h  [{\bf e}_h -  {\bf e}_r]^\top.
\label{Abar}
\ee
In the specific case when the graph is undirected, the disconnection of the arc $(r,h)$ implies also the disconnection of the arc $(h,r)$. Consequently
\begin{eqnarray}
\bar {\mathcal L}  %&=& {\mathcal L} + \ell_{hr} [{\bf e}_h{\bf e}_h^\top + {\bf e}_r{\bf e}_r^\top] - \ell_{hr} [{\bf e}_h{\bf e}_r^\top + {\bf e}_r{\bf e}_h^\top]\\
&=& {\mathcal L} + \ell_{hr} [{\bf e}_h -{\bf e}_r][{\bf e}_h-{\bf e}_r]^\top,\nonumber \\
\bar A  &=& A - \kappa  \ell_{hr} [{\bf e}_h -{\bf e}_r][{\bf e}_h-{\bf e}_r]^\top.
\label{Abar_und}
\end{eqnarray}
% {\color{black} It is quite intuitive that if the edge disconnection
%   compromises the agents' mutual exchange of information
%   {\color{verde} to the extent of destroying graph connectivity}, then
%   consensus in general cannot be achieved after the link failure. So,
%   in the following, we will investigate the effects of edge
%   disconnection by assuming that the strong connectedness of the
%   communication graph is preserved after the failure.}
 If the edge disconnection
  compromises the agents' mutual exchange of information,
 to the extent of destroying the graph connectivity, then
{\color{black} consensus will not be reached, and the effects of the fault on the network  will make fault detection
 eventually possible.
% cannot be recovered by a proper reorganization of the network, but
% only by physically restoring the broken link. 
 %{\color{black} On the other hand, the fact that the system does not achieve consensus would clearly indicate that a fault has occurred.}
On the other hand, an edge disconnection that does not affect 
the graph connectivity, will allow the system to still reach consensus but on a different value from the original one, and the fault may hence go unnoticed and seriously affect the system functioning.}
 For these reasons,
  in this paper we will investigate the effects of a single edge
  disconnection by assuming that the strong connectedness of the
  communication graph is preserved after the failure. 
%We will not address here the case when multiple edges get disconnected at the same time.
{\color{black} The following proposition, that makes use of some results derived in  \cite{OF-Murray2004} for consensus networks with switching topologies,  highlights that under this assumption we still have a consensus network, but in general the consensus value is preserved after the disconnection only if the communication graph is undirected.}
%  Note that the case when a single node gets disconnected, and hence all its edges get disconnected at the same time, has been  investigated in \cite{BattistelliTesi,ICSTCC2019}.

\begin{proposition} \label{uno}  
% {\color{verde} Io sono convinto che il revisore si riferisca al fatto
% che in quegli anni e' stato studiato, ed e' stato provato che, 
% sistemi multiagente con grafo diretto  ma time-varying convergono
% comunque al consensus inteso come funzione lineare delle
%condizioni iniziali (se valgono le condizioni sullo spanning
% tree). In basso i 'milestone papers' di quel periodo (secondo la mia conosenza..).
%REACHING A CONSENSUS IN A DYNAMICALLY CHANGING ENVIRONMENT: CONVERGENCE RATES, MEASUREMENT DELAYS, AND ASYNCHRONOUS EVENTS∗MING CAO†, A. STEPHEN MORSE†, AND BRIAN D. O. ANDERSON‡
%Stability of Multiagent Systems With Time-Dependent Communication Links Luc Moreau
%Consensus Seeking in Multiagent Systems Under Dynamically Changing
%Interaction Topologies Wei Ren and Randal W. Beard
%}
Let ${\mathcal L}$ be the Laplacian of a strongly connected graph ${\mathcal G}$, and 
{\color{black} set $A :=
 I_N-\kappa {\mathcal L}$, where $\kappa>0$ is a fixed coupling strength, that has been chosen in order to ensure that 
 system \eqref{sana} is a consensus network.
Let $\bar {\mathcal L}$ be the Laplacian of the graph $\bar {\mathcal G}$, obtained from ${\mathcal G}$ by removing the arc $(r,h)$,
 and set $\bar A:= I_N-\kappa \bar {\mathcal L}$.} 
 If $\bar {\mathcal G}$ is still strongly connected, then
 \begin{itemize}
 \item[i)] the system ${\bf x}(t+1)=(I_N -\kappa \bar {\mathcal L}) {\bf x}(t) = \bar A {\bf x}(t)$ is still a consensus network;
 \item[ii)]  if the graph ${\mathcal G}$  is undirected, then for every choice of ${\bf x}(0)$ and every time $\tau\ge 0$   at which the edge disconnection may occur, the new network converges to the same consensus value to which the original network would have converged before the disconnection;
 \item[iii)]  if the graph ${\mathcal G}$  is directed,  %and {\color{black} the arc $(r,h)$ is disconnected,} % (while the arc $(h,r)$ remains)
  then
 for every $\tau\ge 0$ there are initial states ${\bf x}(0)$ corresponding to which  the consensus value obtained by the new network, after the  disconnection at $t=\tau$,  differs from the original one \eqref{cons_value}.
 \end{itemize}
\end{proposition}

\begin{proof}
i) \ Since the original network is a consensus network, ${\mathcal L}$ is irreducible and $\kappa$ satisfies the constraint
\eqref{constr_kappa}. On the other hand, by assumption, the Laplacian $\bar {\mathcal L}$ is still irreducible and   if we denote by $\bar\ell_{ij}$ the $(i,j)$-th entry of $\bar {\mathcal L}$, then 
$\max_{i\in[1,N]}\bar\ell_{ii} \le \max_{i\in[1,N]} \ell_{ii},$ thus ensuring that
$$0 < \kappa < \frac{1}{\max_{i\in[1,N]} \bar \ell_{ii}}.$$
This is the case for both directed  and   undirected graphs. Consequently, also the new network is a consensus network.
\smallskip

\noindent {\color{black}ii) and iii) Follow from the results  obtained  in \cite{OF-Murray2004} (see for instance Theorems 4 and 9)
  for the consensus of continuous-time systems described as integrators and with switching communication topologies. }
\end{proof}
\smallskip

%\begin{remark} IO PENSAVO DI TOGLIERLO - SE VUOI RIMETTILO
%Point ii) of Proposition \ref{uno} states a rather strong result, namely that in the case of symmetric communication exchange, as far as the fault 
%does not split the agents into two disjoint groups, the algorithm
%still works and leads to the same consensus value it would have
%achieved before the fault. 
%The symmetric case has been investigated in \cite{ECC2019}. See also the survey paper 
%\cite{ICSTCC2019}.
%\end{remark}

%\noindent {\bf Assumption 2.} In the rest of the paper we will assume that the Laplacian $\bar {\mathcal L}$ (and hence the state matrix $\bar A$) of
%the network, obtained from \eqref{sana} upon disconnection
%of the edge between agent  $r$ and agent $h$, is irreducible.
%\smallskip

\section{Detecting an edge disconnection}
\label{sec4}

 In the rest of the paper we will focus on the case when the 
communication graph is directed, and
 investigate in detail under what conditions we can detect the edge disconnection. We will also steadily make the following  assumption.
 \smallskip
 
\noindent {\bf Assumption 2.}  {\em The graph $\bar{\mathcal G}$, describing the communication network after the link failure,
   is strongly connected, and hence $\bar A$ is still a positive irreducible matrix having $1$ as dominant eigenvalue and ${\bf 1}_N$ as dominant eigenvector.}
\smallskip

We distinguish the case when we can observe the states of all the agents and the case when we can observe the states of a subset of the agents.
We also adjust the definitions introduced in \cite{BattistelliTesi} (for the continuous-time case), 
  to keep into account that 
if the system has already reached consensus,   which is an equilibrium point for both the original network and the faulty one,  then every edge disconnection will not alter the consensus status, and hence will produce a fault that is necessarily undetectable. Consequently,
% (see also Proposition \ref{uno}), 
we introduce the following definitions.
\smallskip

\begin{definition} \label{def1}
Consider the multi-agent consensus network \eqref{sana}, and the network obtained from \eqref{sana} upon disconnection
of the edge from agent  $r$ to agent $h$:
\be\label{rotta}
{\bf x}(t+1)= \bar A {\bf x}(t),
\ee
with $\bar A$ described as in \eqref{Abar}.
The two networks are said to be {\em discernible} %, equivalently the {\em edge disconnection is reconstructible},
 if for every fault time $\tau \ge 0$ and every  state ${\bf x}(\tau) \not\in \langle {\bf 1}_N\rangle$, there exists $t> \tau$ such that the state trajectory
 of the faulty system \eqref{rotta} at time $t$,  ${\bf x}(t) = \bar A^{t-\tau}{\bf x}(\tau)$, is different from the state trajectory of the original system at time $t$.
 \newline 
 %On the other hand, 
 If only the states of $p< N$ agents are available, and we assume without loss of generality that they are the first $p$ agents,
 we say that the two networks are   {\em discernible from the observation of the first $p$ agents} %, equivalently the {\em edge disconnection is reconstructible},
 if for every fault time $\tau \ge 0$ and every  state ${\bf x}(\tau) \not\in \langle {\bf 1}_N\rangle$,  the first $p$ entries of any state trajectory
 of the faulty system \eqref{rotta} at time $t\ge \tau$,  ${\bf x}(t)
 = \bar A^{t-\tau}\bar{\bf x}_\tau$, are different from the first $p$
 entries of the state trajectory of the original system at time $t$
 for at least one time instant $t$, namely   for every $\bar {\bf
   x}_\tau\in {\mathbb R}^N$ there exists $t\ge \tau$ such that
% {\color{verde}  \marginpar{Il Rev11 contesta questa   come regola di detection. Per me qui e' solo una definizione. Aggiungere commento?} }
\be
\begin{bmatrix}I_p &0\end{bmatrix} \bar A^{t-\tau}\bar {\bf x}_\tau
 \ne \begin{bmatrix}I_p &0\end{bmatrix}   A^{t-\tau}{\bf x}(\tau).
 \label{notdisc_p}\ee
 \end{definition}
 
{\color{black} 
\begin{remark}
If the edge fault is the outcome of an external attack, it is immediate to realise that the concept of discernibility is in perfect  agreement with the property of a cyberphysical system to not be subjected to undetectable attacks, explored  in \cite{pasqualetti2015control}.
Note that the concept of discernibility here adopted, and suitably adapted from the one given in \cite{BattistelliTesi}, is different from the concept of ``detectabilty" (of an edge) adopted in \cite{RahimianCDC2012} that only requires that the state trajectories of the healthy and the faulty systems differ for  at least one choice of the initial condition. \\
Finally, discernibility is introduced here as a system theoretic property with   exact definition and mathematical characterisation. It is clear that in a real-life environment, we need to account for  disturbances and modeling errors, that 
require to modify the previous theoretical condition \eqref{notdisc_p} to introduce a minimal threshold below which the disagreement of the measured output with respect to the expected one is not interpreted as the effect of a fault.
\end{remark}}

\subsection{Discernibility after edge disconnection}

In order to characterize discernibility, we may   exploit the analysis in \cite{BattistelliTesi}, that refers to the matrices 
 \be
 \Delta := \begin{bmatrix} A & 0\cr 0 & \bar A\end{bmatrix}
 = I_{2N} - \kappa \begin{bmatrix} {\mathcal L} & 0\cr 0 & \bar {\mathcal L}\end{bmatrix}, \quad
 \Gamma _N:=\begin{bmatrix}I_N & - I_N\end{bmatrix}.
 \label{coppia_oss}
 \ee
It is worth observing, however, that discernibility analysis in \cite{BattistelliTesi} is 
{\color{black} carried out for  homogeneous multi-agent systems, whose agents are generically described by the same  linear state-space model,
%first carried out for generic systems and then assuming {\color{black} as system structure the one obtained by applying} a DeGroot's type consensus law to a homogeneous multi-agent system, whose agents are generically described by the same  linear state-space model.
 and it is defined in such a way that it coincides with} the observability property of the pair $(\Delta, \Gamma_N)$.
% and $(\Delta, \Gamma_p)$.
Observability, however, is impossible to guarantee when the matrices $A$ and $\bar A$ are expressed in terms of the Laplacian as in \eqref{sana} and \eqref{Abar_und}, and this is quite reasonable since the lack of observability is related to the fact that if the disconnection happens when the network is already in its steady state, then the fault cannot be detected,  since the constant trajectory $\alpha {\bf 1}_N$ is compatible both with the original network and with the faulty one. 
We have   modified the two definitions of discernibility just to rule out this case, that  is unavoidable and cannot be regarded as a sign of bad performance. Clearly, this will lead to  different characterizations of the two discernibility properties.

We can now provide the following result, 
that adjusts and extends the one given in  Theorem 1   of \cite{BattistelliTesi}.
\smallskip

\begin{proposition} \label{classica}
Given the networks \eqref{sana} and \eqref{rotta}, this latter obtained from the former after the disconnection of the edge $(r,h)$, 
assume that Assumptions 1 and 2 hold. Then the following facts are equivalent:
\begin{itemize}
\item[i)]  the networks \eqref{sana} and \eqref{rotta} are discernible;
\item[ii)] the unobservable states of the pair $(\Delta, \Gamma_N)$ are those in $\langle {\bf 1}_{2N}\rangle$  and they   correspond to the unitary eigenvalue;
\item[iii)] the unobservable states of the pair $(A, [{\bf e}_r-{\bf e}_h]^\top)$  are those in $\langle {\bf 1}_{N}\rangle$  and they   correspond to the unitary eigenvalue;
\item[iv)]
there is no eigenvalue-eigenvector pair $(\lambda, {\bf v})$, % {\color{black} of the matrix $A$}, 
with $\lambda\in {\mathbb C}$ and ${\bf v}\ne 0$,  except for $\lambda= 1$ and ${\bf v}\in \langle {\bf 1}_N\rangle$, such that
\be
A {\bf v} =\lambda {\bf v} \quad {\rm and}\quad [{\bf v}]_r=[{\bf v}]_h.
\label{cond_not_disc}
\ee
\item[v)]
there is no eigenvalue-eigenvector pair $(\lambda, {\bf v})$, % {\color{black} of the matrix $A$}, 
with $\lambda\in {\mathbb C}$ and ${\bf v}\ne 0$,  common to $A$ and $\bar A$, except for $\lambda= 1$ and ${\bf v}\in \langle {\bf 1}_N\rangle$.

\end{itemize}
\end{proposition}

\begin{proof}  i) $\Leftrightarrow$ ii) \ Suppose that  the networks \eqref{sana} and \eqref{rotta} are not discernible. Then there exist  ${\bf x}(0)\not\in \langle{\bf 1}_N\rangle$  such that
$\bar A^{t}{\bf x}(0)=   A^{t}{\bf x}(0)$ for every $t\ge 0$. This is equivalent to saying that 
$\begin{bmatrix} {\bf x}(0)\cr {\bf x}(0)\end{bmatrix}$ is not observable for the pair $(\Delta, \Gamma_N)$  and   it does not belong to $\langle {\bf 1}_{2N}\rangle$.
\\
Conversely, suppose that there exists an unobservable state of the pair $(\Delta, \Gamma_N)$, ${\bf v}=\begin{bmatrix}{\bf v}_1\cr {\bf v}_2\end{bmatrix} \not \in \langle {\bf 1}_{2N}\rangle$.
Clearly, ${\bf v}_1$ must coincide with ${\bf v}_2$, and hence condition ${\bf v}\not \in \langle {\bf 1}_{2N}\rangle$ implies
${\bf v}_1\not\in \langle {\bf 1}_N\rangle$. This implies that if at some $\tau\ge 0$ the original network gets disconnected when ${\bf x}(\tau)={\bf v}_1$ then 
$\bar A^{t-\tau}{\bf x}(\tau)=   A^{t-\tau}{\bf x}(\tau)$ for every $t\ge \tau$, thus ruling out discernibility.
\newline
ii)\ $\Leftrightarrow$\ iii)\ Condition ii) is easily seen to be equivalent to the following condition, expressed in terms of PBH observabilty matrix:
if there exist $\lambda\in {\mathbb C}$ and $\begin{bmatrix} {\bf v}\cr \bar {\bf v}\end{bmatrix}\ne 0$ such that
\be
\begin{bmatrix} \lambda I_N - A & 0\cr
0 & \lambda I_N -\bar A\cr I_N & - I_N\end{bmatrix} \begin{bmatrix} {\bf v}\cr \bar {\bf v}\end{bmatrix}=0,
\label{T1}\ee
then $\lambda=1$ and $\begin{bmatrix} {\bf v}\cr \bar {\bf v}\end{bmatrix}\in \langle 1_{2N}\rangle$.
Similarly, condition iii) is equivalent to saying that 
if  $\lambda\in {\mathbb C}$ and  ${\bf v} \ne 0$ exist such that
\be
\begin{bmatrix} \lambda I_N - A \cr
 [{\bf e}_h - {\bf e}_r]^\top  \end{bmatrix}   {\bf v}=0, 
\label{T3}\ee
then  $\lambda=1$ and ${\bf v} \in \langle 1_{N}\rangle$.
On the other hand, it is easily seen that \eqref{T1} is equivalent to
\be
\begin{bmatrix} \lambda I_N - A & 0\cr
0 & \lambda I_N -\bar A \end{bmatrix} \begin{bmatrix} {\bf v}\cr   {\bf v}\end{bmatrix}=0,
\label{T2}\ee
namely to 
\begin{eqnarray*}
A {\bf v} &=& \lambda {\bf v},\cr
\bar A {\bf v} &=& \lambda {\bf v},
\end{eqnarray*}
and due to the relation between $A$ and $\bar A$,  the previous two identities are, in turn, equivalent to
\begin{eqnarray*}
A {\bf v} &=& \lambda {\bf v},\cr
[{\bf e}_h - {\bf e}_r]^\top {\bf v} &=& 0,
\end{eqnarray*}
that can be expressed in terms of the PBH observability criterion as in \eqref{T3}.
This proves  ii) $\Leftrightarrow$ iii).

\noindent iii)\ $\Leftrightarrow$ \ iv)\ Obvious.

\noindent iv)\ $\Leftrightarrow$ \ v)\ It is easily seen that   \eqref{cond_not_disc}
holds if and only if 
$A{\bf v}=\lambda {\bf v} = \bar A {\bf v}.$
So, the equivalence immediately follows.
\end{proof}

\begin{remark}  
Conditions ii),  iii) and iv) in Proposition \ref{classica} could be equivalently expressed 
in terms of the matrices ${\mathcal L}$ and $\bar {\mathcal L}$ (instead of $A$ and $\bar A$), and by replacing  $\lambda_A=1$ with $\lambda_{\mathcal L} =0$.
\end{remark}

%\begin{remark} \label{AobarA} Condition iv) in Proposition \ref{classica} could be equivalently expressed by referring to $\bar A$ rather than $A$. In fact, the identity $\bar A = A - \kappa \ell_{hr} {\bf e}_h[{\bf e}_h-{\bf e}_r]^\top$  implies that  a pair $(\lambda, {\bf v})$ satisfies \eqref{cond_not_disc} if and only if satisfies: 
%$\bar A {\bf v} =\lambda {\bf v}$   and  $[{\bf v}]_r=[{\bf v}]_h.$
%\end{remark}

  In the following, we consider some special cases of matrices $A$ (or, equivalently, graph Laplacians ${\mathcal L}$)
   for which condition iv) in Proposition \ref{classica} is violated.  These situations rule out in advance discernibility.
 
%   
% {\color{black} bisogna evidenziare qualcosa di abbastanza ovvio ovvero che gli autovettori di $A$ e del Laplaciano coincidono perche' la proposizione successiva fa riferimento agli autovettori del Laplaciano e non a quelli di $A$. In realta' questo risultato c'e' gia' nel paper di Battistelli e Tesi}

 If there exists $\lambda\in \sigma(A), \lambda\ne 1,$ of geometric multiplicity greater than $1$,
then an eigenvector of $A$ corresponding to $\lambda$ can be found such that condition \eqref{cond_not_disc} is  satisfied thus making the old network and the new network not discernible.
So, a necessary condition for discernibility is that {\color{black}  all the eigenvalues of $A$ have} unitary  geometric multiplicity ($A$
{\color{black} is} cyclic  \cite{SontagBook}
or, equivalently, non-derogatory \cite{Cullen}).
%{\color{black} menzionata nel libro di Sontag "Mathematical Control Theory" ma la definizione e' diversa e non si basa sulla forma di Jordan.
%ETTORE: Fuhrmann in ?A poly approach to Lin Alg? Springer 1996 presenta la forma reale di Jordan.
%Sempre Fuhrmann, e molto pi lungo N.Jacobson in ?Lestures in Abstract Algebra? vol 2 "Linear Algebra? van Nostrand co, 1953 parlano di cyclic linear transformation (e una matrice ciclica rappresenta appunto una cyc lin tr) come di una trasformazione A tale che, per un opportuno vettore $v, v,Av, ?., A^{n-1}v$ sono una base per lo spazio.  
%anche su Horn Johnson, vol 1, pp 150 e segg}

\begin{lemma} \label{lemma4}
If  $A$ has an eigenvalue $\lambda\ne 1$ of geometric multiplicity greater than $1$,
then  there exists an eigenvector ${\bf v}$ corresponding to $\lambda$ such that condition \eqref{cond_not_disc} holds.
%, and therefore 
%the networks \eqref{sana} and \eqref{rotta}, this latter obtained from the former after the disconnection of the edge $(r,h)$, are  not discernible.
\end{lemma}

\begin{proof} If $\lambda\in \sigma(A), \lambda \ne 1,$ has geometric multiplicity  greater than $1$, 
then there exist $2$ linearly independent eigenvectors, say  ${\bf v}_1$ and ${\bf v}_2$, corresponding to $\lambda$.
Suppose that neither of these eigenvectors has the $r$-th and the $h$-th entries that coincide.
%{\color{black} Assume, for instance, that $r <h$, and}
Introduce the  $2\times 2$ matrix 
$$M_{r,h} :=\begin{bmatrix} [{\bf v}_1]_r & [{\bf v}_2]_r\cr
[{\bf v}_1]_h & [{\bf v}_2]_h \end{bmatrix}.$$
If $M_{r,h}$ 
is nonsingular,
  there exist $a_1,a_2\in {\mathbb R}\setminus \{0\}$ such that
$$\begin{bmatrix}1 \cr 1\end{bmatrix} = M_{r,h} \begin{bmatrix} a_1\cr a_2\end{bmatrix}.$$
If $M_{r,h}$ is singular,   there exist  $a_1,a_2\in {\mathbb R}\setminus \{0\}$ such that
$$\begin{bmatrix}0 \cr 0\end{bmatrix} = M_{r,h} \begin{bmatrix} a_1\cr a_2\end{bmatrix}.$$
In both cases the eigenvector ${\bf v} := a_1 {\bf v}_1 + a_2 {\bf v}_2$ satisfies 
$[{\bf v}]_r=[{\bf v}]_h.$
%
%On the other hand, if there exist two eigenvectors ${\bf v}_i$ and ${\bf v}_j$ corresponding to the same eigenvalue then it is always replace those two eigenvectors 
%$\bar {\bf v}_i$ and $\bar {\bf v}_j$ such that $\langle \bar {\bf v}_i, \bar {\bf v}_j\rangle = \langle {\bf v}_i, {\bf v}_j\rangle$ and $[\bar {\bf v}_i]_h =
%[\bar {\bf v}_i]_r$.
%This proves that discernibility is lost if and only if \eqref{cond_not_discA} holds. \hfill$\clubsuit$\end{proof}
%
 %\hfill$\clubsuit$
 \end{proof}
\smallskip

\begin{remark} If the graph ${\mathcal G}$ is undirected,
  $A$ is symmetric and hence diagonalizable. 
  Therefore algebraic multiplicities and geometric multiplicities coincide. So, a necessary condition for discernibility is that  the $N$ eigenvalues of $A$ are all distinct.
 \end{remark}

We now further explore condition \eqref{cond_not_disc} of
  Proposition \ref{classica} and   connect it to a topological condition on ${\mathcal
    G}$. To this end, we need to introduce the concept of nontrivial almost equitable
  partition for a directed
  weighted graph   ${\mathcal G}$, by extending the analogous notion given for undirected unweighted graphs in
 \cite{egerstedt2012interacting}.
%{\color{black} Interacting with networks, by Magnus  see Eqns. (9) and (10)}
  Given a directed weighted graph ${\mathcal G}=({\mathcal V},{\mathcal E}, {\color{black}{\mathcal W}})$, a
partition $\pi=\{ {\mathcal V}_1, .., {\mathcal V}_k\}$ of the set of
vertices ${\mathcal V} = [1,N]$ is said to 
be an   \emph{equitable partition} for ${\mathcal G}$ if 
for every pair of cells ${\mathcal V}_i, {\mathcal V}_j, i,j\in [1,k]$,  and every node $v$ of   ${\mathcal V}_i$,
the sum of the weights of all
 edges   from the  nodes in ${\mathcal V}_j$  
to the node $v$ 
% does 
% not depend on the vertex itself $v$ but it
is a constant value that depends only on $i$ and $j$, not on $v$. The partition $\pi$ is said to be an 
\emph{almost}
 (or   \emph{relaxed}) {\em equitable partition}
 if the above condition holds for every pair $(i,j), i,j\in [1,k]$, with
$j \neq i$. 
%To mathematically formalize this concept,
%we introduce  the \emph{cell-to-node strength/flow/degree}: 
%$$\Phi_\pi(v; {\mathcal
%  V}_j):=  \sum_{u\in {\mathcal V}_j} [{\color{black}{\mathcal W}}]_{vu}, = -  \sum_{u\in {\mathcal V}_j}\ell_{vu},$$ 
% where $\ell_{vu}$ is the $(u,v)$-th entry of the Laplacian associated with ${\color{black}{\mathcal W}}$.
%A partition $\pi$ is an equitable partition for  the directed weighted graph ${\mathcal G}$ 
%if for every ${\mathcal V}_i, {\mathcal V}_j, i,j\in [1,k]$,  and every pair of nodes $v_1,v_2$ of   ${\mathcal V}_i$,
%\be\label{epcondition} \Phi_\pi(v_1; {\mathcal
%  V}_j) = \Phi_\pi(v_2; {\mathcal
%  V}_j). \ee 
%This amounts to saying that $\sum_{u\in {\mathcal
%    V}_j}\ell_{vu}=d_{ij} \le 0$ for every $v\in {\mathcal V}_i$.
%    A partition $\pi$ is an almost  equitable partition if the previous conditions hold for $i\ne j$.
%Note that  for any directed weighted graph ${\mathcal G}$,  two trivial equitable partitions always exist, namely (1) 
%the one corresponding to $k=1$ and ${\mathcal V}_1={\mathcal V}$,
%and (2) the one corresponding to  $k=N$ and each ${\mathcal V}_i$ consists of a single distinct node.
%In the following, when talking about (almost) equitable partitions, we will always mean the not trivial ones.
  
 In formal terms, 
% To mathematically formalize this concept,
%we introduce  the \emph{cell-to-node strength/flow/degree}: 
%$$\Phi_\pi(v; {\mathcal
%  V}_j):=  \sum_{u\in {\mathcal V}_j} [{\color{black}{\mathcal W}}]_{vu}, = -  \sum_{u\in {\mathcal V}_j}\ell_{vu},$$ 
% where $\ell_{vu}$ is the $(u,v)$-th entry of the Laplacian associated with ${\color{black}{\mathcal W}}$.
a partition $\pi$ is an almost equitable partition for  the directed weighted graph ${\mathcal G}$ 
if for every ${\mathcal V}_i, {\mathcal V}_j, i,j\in [1,k], i\ne j$,  and every pair of nodes $v_1,v_2$ of   ${\mathcal V}_i$,
$$
\sum_{u\in {\mathcal V}_j} [{\color{black}{\mathcal W}}]_{v_1u} = \sum_{u\in {\mathcal V}_j} [{\color{black}{\mathcal W}}]_{v_2u},$$
or, equivalently, by using the Laplacian entries
\be\label{epcondition} 
\sum_{u\in {\mathcal V}_j} \ell_{v_1u} = \sum_{u\in {\mathcal V}_j} \ell_{v_2u}. \ee 
This amounts to saying that $\sum_{u\in {\mathcal
    V}_j}\ell_{vu}=d_{ij} $ for every $v\in {\mathcal V}_i$ and every $j\ne i$. Clearly, $d_{ij} <0$.
 %   A partition $\pi$ is an almost  equitable partition if the previous conditions hold for $i\ne j$.
Note that  for any directed weighted graph ${\mathcal G}$,  two trivial almost equitable partitions always exist, namely (1) 
the one corresponding to $k=1$ and ${\mathcal V}_1={\mathcal V}$,
and (2) the one corresponding to  $k=N$ and each ${\mathcal V}_i$ consisting of a single distinct node.
In the following, when talking about almost equitable partitions, we will always rule out the two trivial ones.

The fact that   $\sum_{u\in {\mathcal V}_j}\ell_{vu}=d_{ij}$ for any $v\in {\mathcal V}_i$ and $j\ne i$ suggests that it
is possible to define (an adjacency matrix and hence) a Laplacian ${\mathcal L}_\pi\in{\mathbb R}^{k\times k}$ 
for ${\mathcal G}$, associated with $\pi$,
as follows  \cite{cardoso2007laplacian}:
\be\label{Lpi} 
[{\mathcal L}_\pi]_{ij}=\left\{
  \begin{array}{ccc}
    d_{ij}, &  {\rm if } & i \neq j; \\
 -\sum_{h\in[1,k] \atop h\neq
  i} d_{ih}, &  {\rm if } &  i = j.
  \end{array}\right.
\ee
% $[L_\pi]_{ij}=-d_{ij}$ if $i \neq j$ and $[L_\pi]_{ii}=\sum_{j\neq
%   i} d_{ij}$ otherwise.
We now provide the following result that extends the analogous one
 for
 undirected
 unweighted  graphs derived in  \cite{cardoso2007laplacian}.
%
%{\color{black} inspired by Propositions 1,2, 3 of 'Laplacian.eps' }

\begin{proposition}
\label{PropEP}
Given a directed weighted graph ${\mathcal G}=({\mathcal V},{\mathcal E}, {\color{black}{\mathcal W}})$, let
$\pi=\{ {\mathcal V}_1, .., {\mathcal V}_k\}$ be a partition  of the set of vertices ${\mathcal V}$
 and let  $P_\pi$ be the characteristic matrix of $\pi$ (see Notation in Section I). If
$\pi$ is a (nontrivial) almost equitable partition, then:
\begin{itemize}
\item[i)]  ${\mathcal L}P_{\pi}=P_{\pi}{\mathcal L}_\pi;$
%\be \label{divisor} {\mathcal L}P_{\pi}=P_{\pi}{\mathcal L}_\pi; \ee
\item[ii)] the spectra of ${\mathcal L}_{\pi}$  and  ${\mathcal L}$
  satisfy  $\sigma({\mathcal L}_\pi)\subset \sigma({\mathcal L})$ and
  the associated eigenvectors are related as follows \be  {\mathbf u}\in  {\rm
    ker}(\lambda I_k -{\mathcal L}_\pi) \Rightarrow P _\pi   {\mathbf u}\in  {\rm
    ker}(\lambda I_N-{\mathcal L}); \ee
\item[iii)] $ \forall \lambda\in \sigma({\mathcal L}_{\pi})\subset \sigma({\mathcal L})$, $\exists
  {\mathbf v}\in  {\rm
    ker}(\lambda I_N - {\mathcal L})$ such that $\forall j\in[1,k]$
\be
[ {\mathbf v}]_r=[ {\mathbf v}]_s \qquad \forall r,s\in {\mathcal V}_j.
\label{equalcomp}
\ee
\end{itemize}
\end{proposition}

\begin{proof}  i)\ Set $n_i:= |{\mathcal V}_i|, i\in [1,k].$ It entails no loss of generality assuming 
that ${\mathcal V}_1=[1, n_1]$ and ${\mathcal V}_i=[\sum_{h=1}^{i-1} n_h + 1, \sum_{h=1}^{i-1} n_h + n_i]$ for $i\in [2,k]$. 
Consequently,
$$P_{\pi} = \begin{bmatrix} {\bf 1}_{n_1} & & & \cr
& {\bf 1}_{n_2} & & \cr
& & \ddots & \cr &&& {\bf 1}_{n_k}\end{bmatrix}.$$
It is easily seen that
{\footnotesize$${\mathcal L}P_{\pi}=
\begin{bmatrix}
- \sum_{j\ne 1} d_{1j} {\bf 1}_{n_1}& d_{12} {\bf 1}_{n_1}& \dots & d_{1k}{\bf 1}_{n_1}\cr
d_{21} {\bf 1}_{n_2} & - \sum_{j\ne 2} d_{2j} {\bf 1}_{n_2}& \dots & d_{2k} {\bf 1}_{n_2} \cr 
\vdots & \vdots& \ddots & \vdots \cr
d_{k1} {\bf 1}_{n_k}& d_{k2} {\bf 1}_{n_k}& \dots & - \sum_{j\ne k} d_{kj}{\bf 1}_{n_k} \end{bmatrix}
$$}
where we used the fact that 
if $v\in {\mathcal V}_i$ then 
$0 = {\bf e}_v^\top {\mathcal L} {\bf 1}_N = \sum_{u\in {\mathcal V}_i} \ell_{vu} + \sum_{j\ne i} d_{ij}.$
It is immediate then to see that %\eqref{divisor} 
 i) holds.

\noindent ii)\ Let $\lambda$ be  arbitrary in $\sigma({\mathcal L}_\pi)$. If ${\bf u}$ is an eigenvector of ${\mathcal L}_\pi$ corresponding to $\lambda$, then
${\mathcal L}_{\pi} {\bf u} = \lambda {\bf u}$.  So, by making use of point i),
%\eqref{divisor} 
we get
${\mathcal L} P_\pi  {\bf u} = P_\pi {\mathcal L}_{\pi} {\bf u} = \lambda P_\pi {\bf u}$. This shows that
$\lambda \in \sigma({\mathcal L})$ (and hence 
$\sigma({\mathcal L}_\pi)\subset \sigma({\mathcal L})$),
and that $P_\pi  {\bf u}$ is an eigenvector of ${\mathcal L}$ corresponding to $\lambda$.

\noindent iii)\ By point ii) it is immediate to see that for every eigenvalue  $\lambda$ of ${\mathcal L}_\pi$  
there exists an eigenvector ${\bf v}$ of ${\mathcal L}$ corresponding to $\lambda$, taking the form
${\bf v}= P_\pi {\bf u}$. Such an eigenvector clearly satisfies 
\eqref{equalcomp}.
\end{proof}
%

% It is worth noting that most of the results on equitable partitions
% are derived for unweighted undirected graphs, while equitable
% decompositions are being investigated mathematics

To highlight the impact of this condition on our
objectives, consider Fig.1. As a consequence of Proposition
\ref{PropEP},   the multi-agent system  whose communication graph is depicted in Fig. 1 is
subject to undetectable  edge failures, which are  highlighted  in red. 
Indeed,  according to Proposition \ref{classica} point iv),  the failure of  every link $(r,h)$ such that the matrix $A$  has an eigenvector 
(corresponding to some non-unitary eigenvalue) whose $r$-th and $h$-th
entries  coincide is not  detectable because it produces
a faulty network which is not discernible from the original one.
% If $r$ and $h$ belong to the same $i$-th cell of an almost equitable
% partition then for every nonzero eigenvalue $\lambda\in
% \sigma({\mathcal L}_\pi)\subset \sigma({\mathcal L})$ {\color{verde}***}(and hence
% for every non-unitary eigenvalue of {\color{black}$A = I_N -\kappa
%   {\mathcal L}$})  {\color{verde}***} {\color{verde} La frase fra ** mi sembra non
%   corretta o non precisa. } {\color{verde}(and hence
% for every non-unitary eigenvalue of $A$ related to every nonzero eigenvalue of
% $\mathcal{L}_\pi$ through \eqref{eigAvsL}})  there exists an
% eigenvector of ${\mathcal L}$ (of $A$) whose $r$-th and $h$-th entries coincide.
 If $r$ and $h$ belong to the same $i$-th cell of an almost equitable
partition then it follows directly  from Proposition \ref{PropEP} that
for every non-unitary eigenvalue of $A$,  related  through
\eqref{eigAvsL}   to a nonzero eigenvalue of
$\mathcal{L}_\pi$, there exists an
eigenvector of $A$ whose $r$-th and $h$-th entries coincide.

%This is related to the fact
%that the same partition is not perturbed by any failure of these links so that  Eq.\eqref{divisor} holds for several $\bar{\mathcal L}$ but the
%same $P_{\pi}$ and ${\mathcal L}_\pi $.  
Some of these undetectable
links are critical and their failure may significantly change the
network structure, as for example egde $(1,2)$ whose failure affects the
network strong connectivity.

\begin{figure}\label{EP}
\begin{center}
\hspace{-1cm}
\includegraphics[scale=0.32]{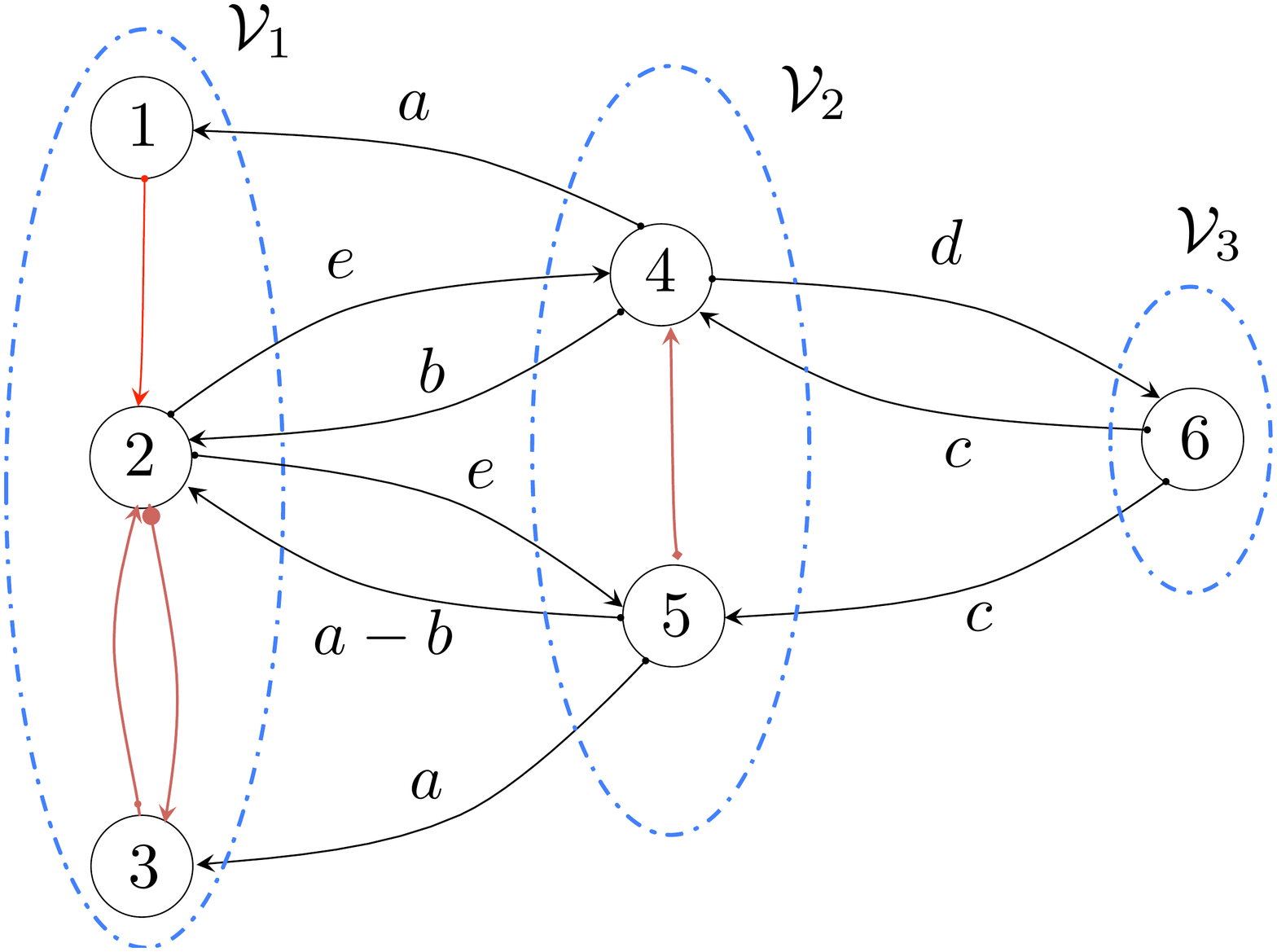}%
%\vspace{-2cm}  
\caption{An example of a graph with a nontrivial equitable partition. Failures of red edges
  are undetectable.}
\end{center}
\end{figure}
%\bigskip

%\begin{remark} As a result of Remark \ref{AobarA}, the previous conditions that rule out discernibility can be expressed also in terms of $\bar A$ and $\bar {\mathcal G}$, namely in terms of the properties of the matrix and the graph of the faulty system.
%\end{remark}

\subsection{How to detect and identify an edge disconnection when the whole state vector is available}

If the states of all the  agents are available, it is possible to easily implement
a residual based fault detection scheme to detect the   disconnection of the edge $(r,h)$. 
Specifically, %by referring to the same notation used within the proof of Proposition \ref{uno},  *****
assume, first,  for the sake of simplicity that all the eigenvalues of $A$ are real and that the Jordan form of $A$ is
 \be\label{Ja} J_A = \begin{bmatrix} 1& 0\cr 0 & \tilde J_A\end{bmatrix} = 
\begin{bmatrix} 1& 0 & \dots & 0 \cr 
0 & J_2   & \dots & 0\cr
\vdots & \vdots& \ddots &\vdots\cr 0 & 0 & \dots  & J_n
\end{bmatrix}, \ee 
where {\color{black} $J_i, i\in [2,n],$ is an elementary Jordan block (of size $k_{i, max}$) corresponding to the eigenvalue $\lambda_i$, with $|\lambda_i|< 1$.
 Note that we do not assume that $\lambda_i\ne \lambda_j$ for $i\ne j$, but each $\lambda_i$ corresponds to a specific chain of generalised eigenvectors ${\bf v}_i^{(k)}, k\in [1, k_{i, max}]$, where  ${\bf v}_i^{(k)},$ is a generalised eigenvector of order $k$ of  $A$ corresponding to  $\lambda_i$.  }
Let $T\in {\mathbb R}^{N\times N}$ be the nonsingular transformation matrix with columns
\be
T =\begin{bmatrix} {\bf 1}_N & {\bf v}_2^{(1)} & {\bf v}_2^{(2)} & \dots & {\bf v}_2^{(k_{2,max})} & \dots & {\bf v}_n^{(k_{n,max})}
\end{bmatrix}.
\label{defT}\ee
Then
$T^{-1} A T = J_A$ and hence $T^{-1} A  = J_A T^{-1}$.  Set
\be
W := \begin{bmatrix} {\bf 0}_{N-1} & I_{N-1}\end{bmatrix} T^{-1}
 \in   {\mathbb R}^{(N-1)\times N}.
\label{defW}
\ee
  Considering \eqref{Ja} and \eqref{defW},  it is easy to verify that $W  A - \tilde J_A W=0$.
% Since the state of the multi-agent system when the system is correctly functioning can be expressed as 
% $${\bf x}(0) = \alpha {\bf 1}_N + \sum_{i=2}^n\sum_{k=1}^{k_{i,max}} \alpha_{ik} {\bf v}_i^{(k)} = T \begin{bmatrix} \alpha\cr \bfalpha_{rem}\end{bmatrix},$$
% where  $\alpha={\bf w}_A^\top {\bf x}(0)$ and 
% $$ \bfalpha_{rem} :=  \begin{bmatrix}  \alpha_{2,1}\cr \alpha_{2,2} \cr \vdots \cr \alpha_{n,k_{n,max}}\end{bmatrix}
% \in {\mathbb R}^{N-1},$$
%then it is easily seen that
%$$W{\bf x}(t) = \tilde J_A^t \bfalpha_{rem}$$
Consequently, as far as the system is correctly functioning, namely  ${\bf x}(t) = A {\bf x}(t-1)$, then the residual signal
$${\bf r}(t) = W {\bf x}(t) - \tilde J_A W {\bf x}(t-1), \quad t\ge 1,$$
is identically zero.  In the case when the matrix $A$ has also complex conjugate
  eigenvalues, then we can follow the same procedure and reasoning as
  above, by replacing the Jordan form with the real Jordan form, and
  hence pairing together pairs of complex conjugate eigenvalues and
  replacing pairs of complex generalised eigenvectors of some order
  with equivalent pairs of real generalised eigenvectors of the same
  order (see e.g.  \cite{HornJohnson}, Section 3.4.1). The details are a little bit more involved from a notational viewpoint, but the substance of the result does not change. For this reason, we omit the details.

We want to show, now, that under the discernibility assumption, unless ${\bf x}(\tau) \in \langle {\bf 1}_N\rangle$, namely the disconnection 
of the edge $(r,h)$ takes place at a time $t=\tau$ when the multi-agent system has already reached consensus, then it is not possible that   ${\bf r}(t)$ is identically zero for $t> \tau$.
This allows one to detect the disconnection of the edge $(r,h)$. Additionally, we   propose conditions that ensure that the edge disconnection is not only detected but also identified, namely the specific broken edge can be identified from the sequence of residual vectors.
To this end, it is convenient to replace the original notation $\bar A$ for the state-space matrix after disconnection with a more specific notation that indicates which specific edge got disconnected. Accordingly, we introduce the following notation:
{$$
\bar A_{ij} := A - \kappa \ell_{ij} {\bf e}_i[{\bf e}_i-{\bf e}_j]^\top,$$
 which is the state matrix of the system once the
edge $(j,i)$ gets disconnected for some $i,j\in [1,N], i\ne j$, (in particular, for $(j,i)=(r,h)$).
%To prove this we need a preliminary lemma.*
\smallskip

\begin{proposition}\label{faultdet_works}
Consider the networks \eqref{sana} and \eqref{rotta}, this latter
obtained from the former after the disconnection of the edge $(r,h)$
at some time $t=\tau$. If Assumptions 1-2 hold, and the networks \eqref{sana} and \eqref{rotta} are discernible (i.e., one of the equivalent conditions of Proposition \ref{classica} holds), then, unless ${\bf x}(\tau)\in \langle {\bf 1}_N\rangle$, namely the network has already reached consensus, there exists $t\in [\tau+1,\tau+N]$ such that ${\bf r}(t)\ne 0$.
\\
Moreover, 
if for every $j\in [1,N]\setminus\{r,h\}$ (and not only for $j=r$), the following conditions hold:\\
(i)\ the faulty network obtained by disconnecting  $(j,h)$ is still strongly connected and discernible   from the original network, \\
(ii)\ 
\be
\sigma(\bar A_{hr})\cap \sigma(\bar A_{hj})=\{1\},\label{disj_spectra}
\ee
%where 
%$$\bar A_{hr} := A - \kappa \ell_{hr} {\bf e}_h[{\bf e}_h-{\bf e}_r]^\top
%\
%{\rm and}\
%\bar A_{hj} := A - \kappa \ell_{hj} {\bf e}_h[{\bf e}_h-{\bf e}_j]^\top,$$
then 
it is possible to identify  from the residual signal the edge $(r,h)$ that got disconnected.
  \end{proposition}
  
\begin{proof}  
If   the   disconnection of edge $(r,h)$
takes place at $t=\tau$, then for every $k\ge 1$
\begin{eqnarray}
{\bf r}(\tau+k) &=& W\bar A_{hr} {\bf x}(\tau+k-1)- \tilde J_A W {\bf
                    x}(\tau+k-1) \nonumber \\
&=& [W \bar A_{hr}  - \tilde J_A W] \bar A_{hr}^{k-1} {\bf x}(\tau)\nonumber\\
&=& - \kappa \ell_{hr}   W {\bf e}_h [{\bf e}_h-{\bf e}_r]^\top \bar
    A_{hr}^{k-1} {\bf x}(\tau), \label{residuo}
\end{eqnarray}
where we used the identities $\bar A_{hr}=A -  \kappa \ell_{hr}    {\bf e}_h [{\bf e}_h-{\bf e}_r]^\top$ and $WA-\tilde J_AW=0$.
Therefore ${\bf r}(\tau+k) =0$ for every $k\ge 1$ if and only if ${\bf x}(\tau)\in
{\rm ker} {\mathcal W}_{hr},$
where
\be 
{\mathcal W}_{hr} := 
\begin{bmatrix}
- \kappa \ell_{hr} W {\bf e}_h ({\bf e}_h-{\bf e}_r)^\top\cr
- \kappa \ell_{hr} W {\bf e}_h ({\bf e}_h-{\bf e}_r)^\top \bar A_{hr}\cr
\vdots\cr
- \kappa \ell_{hr} W {\bf e}_h ({\bf e}_h-{\bf e}_r)^\top \bar A_{hr}^{N-1}
\end{bmatrix},
\label{mat_Whr}\ee
%{\color{verde} Dobbiamo mettere l'ipotesi che avvenga un guasto alla volta} {\color{black} Fatto prima}
This amounts to saying that there exist $\lambda\in {\mathbb C}$ and ${\bf v}\ne 0$ such that 
the PBH observability matrix  satisfies
$$\begin{bmatrix} \lambda I_N - \bar A_{hr}\cr 
W {\bf e}_h [{\bf e}_h-{\bf e}_r]^\top\end{bmatrix} {\bf v}=0,$$
but this is easily seen to be equivalent to
the existence of $\lambda\in {\mathbb C}$ and ${\bf v}\ne 0$ such that 
$$\left\{ \begin{matrix} \bar A_{hr} {\bf v}=\lambda {\bf v}\cr
[{\bf e}_h-{\bf e}_r]^\top {\bf v}=0,
\end{matrix} \right.$$
and therefore 
to
the existence of $\lambda\in {\mathbb C}$ and ${\bf v}\ne 0$ such that 
$$\left\{ \begin{matrix} A{\bf v}=\lambda {\bf v}\cr
[{\bf v}]_h=  [{\bf v}]_r.
\end{matrix} \right.$$
Discernibility assumption rules out the possibility that the previous condition holds unless $\lambda=1$ and ${\bf v}\in \langle {\bf 1}_N\rangle$. On the other hand, if the previous condition holds only for $\lambda=1$ and ${\bf v}\in \langle {\bf 1}_N\rangle$, this means that
${\bf x}(\tau)\in \langle {\bf 1}_N\rangle$, namely the disconnection had taken place after the network had reached consensus, a situation in which detection is not possible. This proves that 
there exists $k\in [1,N]$ such that ${\bf r}(\tau+k)\ne 0$.

Now we want to prove that under the assumptions that:  (i) the disconnection of any edge $(j,h)$ results in a new 
 strongly connected network, discernible from the original one, and   (ii) condition \eqref{disj_spectra} holds, it is possible to uniquely identify the broken link from the residuals. 
By the previous part of the proof,  if the disconnection takes place at $t=\tau$ and ${\bf x}(\tau)\not\in \langle{\bf 1}_N\rangle$, then at least one of the values 
${\bf r}(\tau+k), k=1,2,\dots, N,$ must be nonzero. Set
\be
k^* := \min \{k\ge 1: {\bf r}(\tau+k)\ne 0\}.
\label{kappastar}
\ee
Then 
${\bf r}(\tau+k^*) = c_{k^*} \cdot W {\bf e}_h,$
where 
\begin{eqnarray}
c_{k^*} &:=& - \kappa  \ell_{hr}  ({\bf e}_h-{\bf e}_r)^\top \bar A_{hr}^{k^*-1}{\bf x}(\tau)
\label{ckappastar}\\
&=& - \kappa  \ell_{hr} ({\bf e}_h-{\bf e}_r)^\top  {\bf x}(\tau+k^*-1)
\nonumber \\
&=& - \kappa  \ell_{hr} 
\left(  [{\bf x}(\tau+k^*-1)]_h- [{\bf x}(\tau+k^*-1)]_r
\right) 
\ne 0. \nonumber
\end{eqnarray}
By Lemma \ref{tech4}  in the Appendix, we can claim that this vector uniquely identifies the index $h$, namely one of the extremes of the edge that got disconnected.
We now note that 
$$\begin{bmatrix}{\bf r}(\tau+k^*)\cr 
{\bf r}(\tau+k^*+1)\cr
\vdots\cr
{\bf r}(\tau+k^*+2N-1)
\end{bmatrix}=  \begin{bmatrix}
W {\bf e}_h & & &\cr & W{\bf e}_h & &\cr &&\ddots & \cr
&&& W {\bf e}_h\end{bmatrix} 
$$
$$\cdot \begin{bmatrix}
  ({\bf e}_h-{\bf e}_r)^\top\cr
  ({\bf e}_h-{\bf e}_r)^\top \bar A_{hr}\cr
\vdots\cr
 ({\bf e}_h-{\bf e}_r)^\top \bar A_{hr}^{2N-1}
\end{bmatrix} (- \kappa \ell_{hr}) {\bf x}(\tau+k^*-1).$$
We have just proved that  we can uniquely identify $W{\bf e}_h$ from the first nonzero residual.
Moreover, the block diagonal matrix having $W{\bf e}_h$ as diagonal block is clearly of full column rank, and hence the vector
${\color{black} {\bf  Y}}\ne 0$
%$${\bf c} := \begin{bmatrix} {\bf c}_1 \cr {\bf c}_2\cr \vdots\cr {\bf c}_N\end{bmatrix}$$ 
such that 
$$\begin{bmatrix}{\bf r}(\tau+k^*)\cr 
{\bf r}(\tau+k^*+1)\cr
\vdots\cr
{\bf r}(\tau+k^*+2N-1)
\end{bmatrix}=  \begin{bmatrix}
W {\bf e}_h & & &\cr & W{\bf e}_h & &\cr &&\ddots & \cr
&&& W {\bf e}_h\end{bmatrix} {\color{black} {\bf  Y}}$$
is uniquely determined. 
Now, we want to show that under assumptions (i) and (ii) we can uniquely identify the index $r$ such that
$${\color{black} {\bf  Y}}\in {\rm Im} \begin{bmatrix}
  ({\bf e}_h-{\bf e}_r)^\top\cr
  ({\bf e}_h-{\bf e}_r)^\top \bar A_{hr}\cr
\vdots\cr
 ({\bf e}_h-{\bf e}_r)^\top \bar A_{hr}^{2N-1}
\end{bmatrix}.$$
If this were not the case, then there would be another index $j\ne r$ (and $j\ne h$) such that
$${\color{black} {\bf  Y}}\in {\rm Im} \begin{bmatrix}
  ({\bf e}_h-{\bf e}_r)^\top\cr
  ({\bf e}_h-{\bf e}_r)^\top \bar A_{hr}\cr
\vdots\cr
 ({\bf e}_h-{\bf e}_r)^\top \bar A_{hr}^{2N-1}
\end{bmatrix} \cap {\rm Im} \begin{bmatrix}
  ({\bf e}_h-{\bf e}_j)^\top\cr
  ({\bf e}_h-{\bf e}_j)^\top \bar A_{hj}\cr
\vdots\cr
 ({\bf e}_h-{\bf e}_j)^\top \bar A_{hj}^{2N-1}
\end{bmatrix},$$
and hence there would be two nonzero vectors ${\bf z}_r$ and ${\bf z}_j$  such that 
$${\color{black} {\bf  Y}} = \begin{bmatrix}
  ({\bf e}_h-{\bf e}_r)^\top\cr
  ({\bf e}_h-{\bf e}_r)^\top \bar A_{hr}\cr
\vdots\cr
 ({\bf e}_h-{\bf e}_r)^\top \bar A_{hr}^{2N-1}
\end{bmatrix} {\bf z}_r = -  \begin{bmatrix}
  ({\bf e}_h-{\bf e}_j)^\top\cr
  ({\bf e}_h-{\bf e}_j)^\top \bar A_{hj}\cr
\vdots\cr
 ({\bf e}_h-{\bf e}_j)^\top \bar A_{hj}^{2N-1}
\end{bmatrix} {\bf z}_j.$$
%{\color{verde} It is worth remarking that the above relations
%  correspond to a recursive application of \eqref{ident-r} using \eqref{idrtauk}.}
Clearly, neither ${\bf z}_r$ nor ${\bf z}_j$ can belong to $\langle {\bf 1}_N\rangle$, otherwise ${\color{black} {\bf  Y}}$ would be zero.
Condition
$$0= \begin{bmatrix}
  ({\bf e}_h-{\bf e}_r)^\top &  ({\bf e}_h-{\bf e}_j)^\top\cr
  ({\bf e}_h-{\bf e}_r)^\top \bar A_{hr} &  ({\bf e}_h-{\bf e}_j)^\top \bar A_{hj} \cr
\vdots & \vdots \cr
 ({\bf e}_h-{\bf e}_r)^\top \bar A_{hr}^{2N-1} & ({\bf e}_h-{\bf e}_j)^\top \bar A_{hj}^{2N-1}
\end{bmatrix} \begin{bmatrix} {\bf z}_r\cr {\bf z}_j\end{bmatrix}$$
corresponds to saying that the unobservable subspace of the  matrix pair
$$\left(\begin{bmatrix} \bar A_{hr} & 0\cr 0 & \bar A_{hj}\end{bmatrix}, 
\begin{bmatrix}
  ({\bf e}_h-{\bf e}_r)^\top &  ({\bf e}_h-{\bf e}_j)^\top\end{bmatrix}\right)$$
  includes the vector $\begin{bmatrix} {\bf z}_r\cr {\bf z}_j\end{bmatrix}\not\in \langle {\bf 1}_{2N}\rangle.$
  Clearly, by the irreducibility assumption on $\bar A_{hr}$ and $\bar A_{hj}$,  this cannot be an eigenvector  corresponding to $\lambda=1$. On the other hand, the fact that ${\bf z}_r$ and ${\bf z}_j$ are both nonzero implies that there is  an eigenvalue $\lambda \ne 1$ common to $\sigma(\bar A_{hr})$ and $\sigma(\bar A_{hj})$. But this contradicts   assumption (ii), and hence $r$ is uniquely determined.  %\hfill$\clubsuit$
\end{proof}
  
We want now to sketch an algorithm to identify the edge $(r,h)$ that got disconnected.
Suppose that at $t=\tau$ the edge $(r,h)$ gets disconnected and that 
 the first nonzero residual after $t=\tau$ is
${\bf r}(\tau+k^*) = c_{k^*} \cdot W {\bf e}_h$, with $k^*>0$ and $c_{k^*}\ne0$ defined as in  \eqref{kappastar} and \eqref{ckappastar}, respectively.
 By the previous reasoning, we can claim that there exists a unique value of $h\in [1,N]$ such that 
${\bf r}(\tau+k^*) \in \langle W {\bf e}_h\rangle$, and this allows to uniquely identify $h$ and hence the coefficient $c_{k^*}$.
From the knowledge of $h$ and $c_{k^*}$ one can infer the identity of
$r$ by comparing  the state value of each in-neighbour of $h$ with the
state  value
of the same node  which produces the residual ${\bf
  r}(\tau+k^*)$. Since it must hold
\be\label{idrtauk} [{\bf x}(\tau+k^*-1)]_r =\dfrac{c_{k^*}}{\kappa  \ell_{hr} }  + [{\bf x}(\tau+k-1)]_h, \ee
 $r$ must belong to the following set}
\begin{eqnarray*}
r\in {\mathcal R}_{k^*} \!\!\!&:=&\!\!\! \{i\in [1,N], i\ne h: \ell_{hi}\ne 0 \ {\rm and}\\
&&\!\!\! [{\bf x}(\tau+k^*-1)]_i  
=\dfrac{c_{k^*}}{\kappa  \ell_{hi} }  + [{\bf x}(\tau+k^*-1)]_h\}.
\end{eqnarray*}
If $|{\mathcal R}_{k^*}|=1$, then $r$ is identified at the first step. If not, one can evaluate the set ${\mathcal R}_{k^*+1}$
and then the intersection ${\mathcal R}_{k^*}  \cap {\mathcal R}_{k^*+1}$. By proceeding in this way, based on the previous proof, this procedure  identifies in a finite number of steps the  value of the index $r$ that represents the first extreme of the edge that got disconnected, since there exists {\color{black}$0\le d \le 2N-1$} such that the set $\cap_{i=0}^d {\mathcal R}_{k^*+i}$ consists of a single element.
%{\color{purple} commento: la detection avviene a $\tau=k^*$ e quindi al piu' in $N$ passi, invece l'identification richiede $k^*+2N-1$ passi... nel worst case }
\smallskip

We now explore the more interesting case of discernibility from the observation of the first $p$ agents.

\section{Discernibility from the observation of the first $p$ agents after edge disconnection}

%\noindent {\bf Assumption 3.} In the following we will assume that the pair $(A, \begin{bmatrix} I_p &0\end{bmatrix})$ is observable. Differently,
%it would be impossible to have discernibility from the observation of the first $p$ agents.
%%be able to deduce sufficient information on the multi-agent system dynamics from  the states of the first $p$ agents.
%\smallskip
  By referring to the second part of Definition \ref{def1}, it is easily seen that discernibility of the two
 systems from the observation of the first $p$ agents imposes the
 observability of the original system. 
 If not, condition \eqref{notdisc_p} would be contradicted for any unobservable state ${\bf x}(\tau)$ and ${\bf x}_\tau=0$.
 \\
 On the other hand, the lack of observability of the faulty system could lead to some pathological situations, since the output measurements could possibly lead to believe that the faulty network has already reached the consensus to some constant value, while it is still evolving.  So, in the following we will assume:
\smallskip

{\bf Assumption 3.} Both the original system and the faulty one are observable from the first $p$ agents, namely both $(A, \begin{bmatrix} I_p&0\end{bmatrix})$  and $(\bar A_{hr}, \begin{bmatrix} I_p&0\end{bmatrix})$ are observable.
\smallskip

Under this assumption, 
we will characterise discernibility from the observation of the first $p$ agents
in terms of the   matrix pair $(\Delta, \Gamma_p)$, with
 \be
 \Delta := \begin{bmatrix} A & 0\cr 0 & \bar A_{hr}\end{bmatrix}\qquad
% = I_{2N} - \kappa \begin{bmatrix} {\mathcal L} & 0\cr 0 & \bar {\mathcal L}\end{bmatrix}, \
 \Gamma_p :=\begin{bmatrix}I_p &0 & - I_p&0\end{bmatrix}.
 \label{coppia_oss_p}
 \ee
 %{\color{black} vero? a differenza dell'osservabilta' qui non ho che lo stato iniziale e' comunque fissato? non e' che devo imporre che gli unici stati iniziali non osservabili con entrambi i blocchi UGUALI siano quelli in $\langle {\bf 1}_{2N} \rangle$ ma la mancanza di osservabilita' non e' il problema corretto?}
It is worth noticing that since $A$ is a positive irreducible matrix, having ${\bf 1}_N$ as dominant eigenvector corresponding to the unitary eigenvalue, clearly $1$ is always an observable eigenvalue 
of the pair $(A, \begin{bmatrix} I_p&0\end{bmatrix})$, and hence if the pair would not be observable, the   eigenvalues of the non-observable subsystem would necessarily have modulus smaller than $1$.
The same reasoning applies to  $\bar A_{hr}$, as far as  it remains irreducible.
Finally, the irreducibility assumption on both $A$ and $\bar A_{hr}$ ensures that the eigenspace of both $A$ and $\bar A_{hr}$ corresponding to $\lambda=1$ is $\langle {\bf 1}_N\rangle$. So, 
the only unobservable eigenvectors of $(\Delta, \Gamma_p)$ corresponding to the unitary eigenvalue are those belonging to $\langle {\bf 1}_{2N}\rangle$.
%As a result, the case $\lambda=1$ does not require any  check. One only needs to evaluate what happens of the PBH observability matrix when $\lambda\ne 1$. As mentioned before, the observability of the original network \eqref{sana} follows from the definition of 
%``discernibility from  the observation of the first $p$ agents" itself. We will introduce the additional assumption, as in 
% \cite{BattistelliTesi}, that also the faulty system is observable, to avoid that  a faulty  network which is dynamically evolving 
% is erroneously regarded as a non-faulty network that has reached consensus.

\noindent Assumption 3 and the previous comments are fundamental to derive the following result, that   extends 
   Proposition 3 in \cite{BattistelliTesi}.
   % we can obtain the following result, whose proof is very similar to the proof of Proposition \ref{classica}.
\smallskip

\begin{proposition}  \label{ostrega}   
 Consider the networks \eqref{sana} and \eqref{rotta}, this latter obtained from the former after the disconnection of the edge $(r,h)$, and assume that Assumptions 1, 2 and 3 hold. The following facts are equivalent:
\begin{itemize}
\item[i)] the networks \eqref{sana} and \eqref{rotta} are discernible from  the observation of the first $p$ agents;
\item[ii)] the unobservable states of the pair $(\Delta, \Gamma_p)$ are those in $\langle {\bf 1}_{2N}\rangle$ and they   correspond to the unitary eigenvalue;
\item[iii)] for every $\lambda \in \sigma(A)\cap\sigma(\bar A_{hr}), \lambda\ne 1$,
$${\rm rank}\begin{bmatrix}\lambda I_N - A& 0\cr 0 & \lambda I_N -\bar A_{hr}\cr
I_p\ \ \ 0 & - I_p \ \ \ 0\end{bmatrix}= 2N;
$$
\item[iv)]
there are no $\lambda\in {\mathbb C}$   and nonzero vectors ${\bf v}, \bar {\bf v}$, except for  $\lambda=1$ and ${\bf v}=\bar {\bf v}\in \langle {\bf 1}_N \rangle$, such that
\be
\left\{\begin{array}{rcl}
A {\bf v} =\lambda {\bf v},  && 
\bar A_{hr} \bar {\bf v} = \lambda \bar {\bf v} \\ 
\begin{bmatrix}I_p&0\end{bmatrix} {\bf v} &=& 
\begin{bmatrix}I_p&0\end{bmatrix} \bar {\bf v}.
\end{array}\right.
\label{cond_not_disc_pp}
\ee
\end{itemize}
\end{proposition}
 \smallskip
 
\begin{proof}  i) $\Leftrightarrow$ ii) \  Suppose that  the networks \eqref{sana} and \eqref{rotta} are not discernible from  the observation of the first $p$ agents. Then there exist  ${\bf x}(0)\not\in \langle{\bf 1}_N\rangle$  and $\bar {\bf x}_0\in {\mathbb R}^N$ such that
$\begin{bmatrix}I_p&0\end{bmatrix} \bar A^{t}{\bf x}_0=   \begin{bmatrix}I_p&0\end{bmatrix} A^{t}{\bf x}(0)$ for every $t\ge 0$. This is equivalent to saying that 
$\begin{bmatrix} {\bf x}(0)\cr {\bf x}_0\end{bmatrix}$, which does not belong to $\langle {\bf 1}_{2N}\rangle$, is not observable for the pair $(\Delta, \Gamma_p)$.
\\
Conversely, suppose that there exists an unobservable state of the pair $(\Delta, \Gamma_p)$, ${\bf x} \not \in \langle {\bf 1}_{2N}\rangle$. Since the only eigenvectors of $\Delta$ corresponding to the unitary eigenvalue and belonging to the  unobservable subspace are those in $\langle {\bf 1}_{2N}\rangle$, this implies that there exists an eigenvector ${\bf v}=\begin{bmatrix}{\bf v}_1\cr {\bf v}_2\end{bmatrix} \not \in \langle {\bf 1}_{2N}\rangle$ of $\Delta$ corresponding to some $\lambda\ne 1$ and  satisfying
$$\begin{bmatrix} I_p&0\end{bmatrix} {\bf v}_1=\begin{bmatrix} I_p&0\end{bmatrix} {\bf v}_2.$$
Note that Assumption 3 ensures that both ${\bf v}_1$ and ${\bf v}_2$ are nonzero vectors.
Since ${\bf v}_1\not\in \langle {\bf 1}_n\rangle$, if $\lambda$ is real then we have found a state that contradicts discernibility from the observation of the first $p$ agents. If $\lambda$ is complex  then we can simply use the real part of ${\bf v}_1$ to disprove discernibility from the observation of the first $p$ agents.
\newline 
ii)\ $\Leftrightarrow$\ iii)\   Condition ii) is easily seen to be equivalent to the following condition, expressed in terms of PBH observabilty matrix:
if there exist $\lambda\in {\mathbb C}$ and $\begin{bmatrix} {\bf v}\cr \bar {\bf v}\end{bmatrix}\ne 0$ such that
\be
\begin{bmatrix}\lambda I_N - A& 0\cr 0 & \lambda I_N -\bar A_{hr}\cr
I_p\ \ \ 0 & - I_p \ \ \ 0\end{bmatrix}\begin{bmatrix} {\bf v}\cr \bar {\bf v}\end{bmatrix}=0,
\label{T1bis}\ee
then $\lambda=1$ and $\begin{bmatrix} {\bf v}\cr \bar {\bf v}\end{bmatrix}\in \langle 1_{2N}\rangle$.
Clearly any such $\lambda$ must be in $\sigma(\Delta)$. On the other hand,  if $\lambda$ would not be a common {\color{black} eigenvalue} of $A$ and $\bar A_{hr}$ then either ${\bf v}$ or $\bar {\bf v}$ would be zero and this would mean that either 
$(A, \begin{bmatrix} I_p&0\end{bmatrix})$  or $(\bar A_{hr}, \begin{bmatrix} I_p&0\end{bmatrix})$ are not observable. This would contradict Assumption 3. Therefore we have proved that ii) is equivalent to iii).
 
\noindent iii)\ $\Leftrightarrow$ \ iv)\ Obvious.
\end{proof}

\begin{remark}  
%{\color{verde} Confrontare con Remark13-analogo    contenuto.}   
    If the networks \eqref{sana} and \eqref{rotta} are discernible from  the observation of the first $p$ agents, they are discernible. If not,  a state ${\bf x}\not \in \langle {\bf 1}_N\rangle$ could be found such that 
 $\bar A_{hr}^{t}{\bf x}= A^t {\bf x}$ for every $t\ge 0$, and hence a fortiori
  $\begin{bmatrix} I_p&0\end{bmatrix}\bar A_{hr}^{t}{\bf x}= \begin{bmatrix} I_p&0\end{bmatrix}A^t {\bf x}$ for every $t\ge 0$.
This implies that  a necessary condition for discernibility from  the observation of the first $p$ agents is that all the nonunitary eigenvalues of  $A$ and $\bar A_{hr}$ have unitary geometric multiplicity. \end{remark}

\subsection{How to detect and identify an edge disconnection when the states of the first $p$ agents are available}

 Also in this case we may detect an edge disconnection by making use of the measurements of the states of the first $p$ agents. 
Since the pair $(A, \begin{bmatrix} I_p &0\end{bmatrix})$ is observable,  let $L$ be a matrix in ${\mathbb R}^{N\times p}$ such that
$A+L \begin{bmatrix} I_p &0\end{bmatrix}$ is nilpotent.
We can construct the   closed-loop dead-beat observer of the state of the multi-agent system \cite{OReilly} as
\be
\hat {\bf x}(t+1) = A \hat {\bf x}(t) -L [{\bf y}(t) - \begin{bmatrix}
  I_p &0\end{bmatrix}  \hat {\bf x}(t)].
\label{DBO}
\ee
Clearly, after a   finite number of steps {\color{black} $\tau_0$\footnote{{\color{black}By exploiting the observability of the pair $(A, \begin{bmatrix} I_p &0\end{bmatrix})$ and Rosenbrock's theorem, we can claim that the minimum 
$\tau_0$ ranges in the interval 
$\left[ \lceil \frac{N}{p}\rceil, N-p+1\right]$.}}} that depends on the nilpotency index of $A+L\begin{bmatrix} I_p &0\end{bmatrix}$,
we have $\hat {\bf x}(t)={\bf x}(t)$, and hence the residual signal %(which is nothing but the output estimation error)
$${\bf r}(t) =  \begin{bmatrix} I_p &0\end{bmatrix}\hat {\bf x}(t) - {\bf y}(t) =  \begin{bmatrix} I_p &0\end{bmatrix}[\hat {\bf x}(t) - {\bf x}(t)],$$
is identically zero from $t=\tau_0$ onward until a fault occurs.
Now suppose that at $t=\tau\ge \tau_0$ the disconnection of the edge $(r,h)$ takes place\footnote{{\color{black}In the transient phase, until the estimation error goes to zero, the residual signal may be not zero. As a result it is not possible to detect an edge disconnection in a reliable way. This is the reason why a dead-beat observer is preferable over an asymptotic observer, since this transient phase lasts a finite number of time instants.}}, and hence the multi-agent state updates according to \eqref{rotta}.
We want to show that, unless the multi-agent system has already reached consensus,   under Assumptions  1, 2 and 3 and any of the equivalent conditions of Proposition \ref{ostrega}   the residual signal will necessarily become nonzero at some time instant $t> \tau$.
 \\
To this goal it is sufficient to show that for  the system obtained by putting together \eqref{rotta} and \eqref{DBO}, namely
 \begin{eqnarray}
 \!\!\begin{bmatrix} \hat {\bf x}(t+1) \cr {\bf x}(t+1)\end{bmatrix}
 \!\! &=& \!\! \begin{bmatrix} A + L \begin{bmatrix} I_p &0\end{bmatrix} & - L \begin{bmatrix} I_p &0\end{bmatrix}\cr 0 & \bar A_{hr}\end{bmatrix} 
 \begin{bmatrix} \hat {\bf x}(t) \cr {\bf x}(t)\end{bmatrix}\\
 {\bf r}(t) \label{res_gen_error} \!\!&=& \!\! \begin{bmatrix} \begin{bmatrix} I_p
     &0\end{bmatrix} & \begin{bmatrix} - I_p
     &0\end{bmatrix}\end{bmatrix}
  \begin{bmatrix} \hat {\bf x}(t) \cr {\bf x}(t)\end{bmatrix}
 \end{eqnarray}
the only unobservable states are those in $\langle {\bf 1}_{2N}\rangle.$\\
{\color{black} Consider the PBH observability matrix
\be
\begin{bmatrix} \lambda I_N - A - L \begin{bmatrix} I_p &0\end{bmatrix} &  L \begin{bmatrix} I_p &0\end{bmatrix}\cr 0 & \lambda I_N - \bar A_{hr}\cr\hline
 \begin{bmatrix} I_p &0\end{bmatrix} & \begin{bmatrix}- I_p &0\end{bmatrix}\end{bmatrix}.
 \label{nuova}\ee
 By resorting to elementary   operations on the rows of the matrix, 
 it is easily seen that the previous matrix is of full column rank for $\lambda\in {\mathbb C}$ if and only if
  % si prova premoltiplicando per \begin{bmatrix} I_N 0 & L\cr 0 & I_N 0\cr 0 & 0 &I_P\end{bmatrix}
 the 
 PBH observability matrix
\be
\begin{bmatrix} \lambda I_N - A   &  0 \cr 0 & \lambda I_N - \bar A_{hr}\cr\hline
 \begin{bmatrix} I_p &0\end{bmatrix} & \begin{bmatrix} -I_p &0\end{bmatrix}\end{bmatrix} 
 \label{nuovaa}
 \ee
 is of full column rank for that $\lambda$. Moreover when both PBH matrices have not full column rank, they have the same kernel.
 By the assumption that
 \eqref{sana} and \eqref{rotta} are discernible from the observation of the first $p$ 
agents it follows (see condition iii) of Proposition \ref{ostrega}) that \eqref{nuovaa} is of full column rank for every $\lambda\ne1$ and for $\lambda=1$ its kernel is $\langle {\bf 1}_{2N}\rangle.$ Therefore, 
the only unobservable states of the overall system  are those in $\langle {\bf 1}_{2N}\rangle$, and this implies that}
the residual ${\bf r}(t)$ cannot remain zero after an edge disconnection.

\smallskip

% EDGE IDENTIFICATION
 
We now want to show that, under suitable assumptions, one can deduce from the residual the exact information about which edge got disconnected.

\begin{proposition}\label{faultdetp_works}
Consider the networks \eqref{sana} and \eqref{rotta}, this latter obtained from the former after the disconnection of the edge $(r,h)$ at some time {\color{black} $t=\tau\ge \tau_0$}.  Assume that    for every edge $(j,i), \ i,j\in [1,N], j\ne i$, (and not only for $(j,i)=(r,h)$) \\
  (i) the faulty network obtained by disconnecting  $(j,i)$ is strongly connected and discernible   from the original network based on the observation of the first $p$ states,  \\
 (ii) %Assumptions 2 and 
{\color{black} Assumption 3 holds   and hence    $(A, \begin{bmatrix}I_p &0\end{bmatrix})$ and $(\bar A_{ij}, \begin{bmatrix}I_p &0\end{bmatrix})$ are observable},  and   \\ (iii)   \be
\sigma(\bar A_{hr})\cap  \sigma(\bar A_{ij}) =\{1\}. \label{disj_spectra_p}
\ee
%where 
%$$\bar A_{hr} := A - \kappa \ell_{hr} {\bf e}_h[{\bf e}_h-{\bf e}_r]^\top
%\
%{\rm and}\
%\bar A_{hj} := A - \kappa \ell_{hj} {\bf e}_h[{\bf e}_h-{\bf e}_j]^\top,$$
Then,  unless ${\bf x}(\tau)\in \langle {\bf 1}_N\rangle$, namely the network has already reached consensus, 
it is possible to identify  from the residual signal ${\bf r}(t), t\ge \tau$, generated by \eqref{res_gen_error},
the edge $(r,h)$ that got disconnected.
 \end{proposition}

\begin{proof}
   To prove that the previous observer-based residual generator produces distinct residual sequences corresponding to different faulty systems
(provided that ${\bf x}(\tau)$, the  state of the multi-agent system at the time of edge disconnection, is not in the equilibrium, yet, namely it is not a multiple of ${\bf 1}_{N}$), it is sufficient to prove that if $(r,h)\ne (j,i)$ then   the  two systems 
$$\left( {\mathbb A}_{hr}, \begin{bmatrix}I_p&0_{p\times (N-p)} &- I_p&0_{p\times (N-p)}\end{bmatrix}  
\right)$$
$$ \left(  {\mathbb A}_{ij}, \begin{bmatrix}I_p&0_{p\times (N-p)} &- I_p&0_{p\times (N-p)}\end{bmatrix} 
\right),
$$
with
$${\mathbb A}_{hr} := \begin{bmatrix} A_L &- L \begin{bmatrix}I_p&0_{p\times (N-p)} \end{bmatrix} \cr 0 &  {\bar A}_{hr}\end{bmatrix},$$
$${\mathbb A}_{ij} := \begin{bmatrix} A_L & -   L \begin{bmatrix}I_p&0_{p\times (N-p)} \end{bmatrix} \cr 0 & {\bar A}_{ij}\end{bmatrix},$$
generate distinct residual trajectories, provided that neither of them has already reached the equilibrium {\color{black} at the time the disconnection occurs}. This amounts to saying that 
the only unobservable states of the system
\be
\left(
\begin{bmatrix} {\mathbb A}_{hr} &\vline& 0\cr\hline 0 &\vline& {\mathbb A}_{ij}\end{bmatrix}, \begin{bmatrix} I_p&0  &
-I_p&0&\vline \ 
-I_p& 0  & I_p& 0 \end{bmatrix}
\right),
\label{confronto}\ee
%$$\left(
%\begin{bmatrix} {\mathbb A}_{hr} & 0\cr 0 & {\mathbb A}_{ij}\end{bmatrix}, \begin{bmatrix} I_p&0_{p\times (N-p)} &
%I_p&0_{p\times (N-p)} & 
%-I_p& 0_{p\times (N-p)} & -I_p& 0_{p\times (N-p)} \end{bmatrix}
%\right),$$
{\color{black} taking the form $\begin{bmatrix} {\bf v}_{hr}\cr {\bf v}_{hr}\cr {\bf v}_{ij} \cr {\bf v}_{ij}\end{bmatrix}$}   are those belonging  to $\langle  \begin{bmatrix}   {\bf 1}_{2N} \cr 0 \end{bmatrix}, \begin{bmatrix}   0\cr {\bf 1}_{2N}\end{bmatrix}\rangle.$ \\
{\color{black} By making use of the PBH observability matrix:
\be
\begin{bmatrix} \lambda I_{2N} - {\mathbb A}_{hr} &   0\cr 
0 & \lambda I_{2N} - {\mathbb A}_{ij}\cr \hline
[I_p \ 0\ - I_p \ 0] &   [ - I_p \ 0 \ I_p \ 0]\end{bmatrix} 
\label{PBH_lambda}\ee
 it is easily seen that 
a vector with the previous block structure belongs to the kernel of \eqref{PBH_lambda} for $\lambda=1$ if and only if ${\bf v}_{hr}$ is a common eigenvector (corresponding to $\lambda=1$) of $A$ and $\bar A_{hr}$ and ${\bf v}_{ij}$ is a common eigenvector (corresponding to $\lambda=1$) of $A$ and $\bar A_{ij}$. Therefore
the overall vector belongs  to $\langle  \begin{bmatrix}   {\bf 1}_{2N} \cr 0 \end{bmatrix}, \begin{bmatrix}   0\cr {\bf 1}_{2N}\end{bmatrix}\rangle.$ }\\
On the other hand, if $\lambda \in \sigma(\bar A_{hr}), \lambda\ne 1,$ then  it is easy to see that under the hypothesis 
 \eqref{disj_spectra_p} and by the observability of $(A_L, \begin{bmatrix} I_p & 0\end{bmatrix})$, we have 
 $\lambda\not\in \sigma({\mathbb A}_{ij})$, and therefore 
 {\color{black}it must be ${\bf v}_{ij}=0$. But this means that $\begin{bmatrix} {\bf v}_{hr}\cr {\bf v}_{hr}\end{bmatrix}$
 should belong to the kernel of \eqref{nuova}, but for $\lambda\ne 1$ the matrix \eqref{nuova} is of full column rank. Analogous reasoning holds if $\lambda \in \sigma(\bar A_{ij}), \lambda\ne 1.$}
\end{proof}

{\color{black}
\section{An illustrative example}
\label{example}

We now apply the previous results   to the case
of a network of $7$ agents, whose communication graph is depicted in Figure
2
 assuming that each agent is described as a discrete-time integrator and runs the algorithm
  \eqref{cons_protocol} with $\kappa=0.25$. All the weights of
the graph are equal to $1$. The set of the observed
nodes is $\{1, 2, 3\}$, and we apply the strategy described in Section V.
The resulting system matrix is \be\label{Asist}
A=\begin{bmatrix}
  0.75 & 0 & 0 & 0.25 & 0 & 0& 0 \\
  0 & 0.75 & 0 & 0 & 0.25 & 0& 0 \\
  0  & 0 & 0.75 & 0 & 0 & 0.25& 0 \\
  0  & 0 & 0 & 0.75 & 0 & 0& 0.25 \\
  0.25  & 0 & 0    & 0  & 0.5 & 0.25 & 0\\
0 &  0.25  & 0 & 0    & 0  & 0.75 & 0 \\
  0  & 0 & 0.25    & 0  & 0.25 & 0 & 0.5\\
\end{bmatrix}
\ee
and it is easily verified that the pair $(A, \begin{bmatrix} I_3 &
  0\end{bmatrix})$ is observable.
A dead beat state observer has been derived  following three steps:
(1)  the pair $(A,   [ 
 I_3 \ 0] )$, with $A$ as in \eqref{Asist}, %\eqref{sana}
                                %%{\color{red}    va scritta
                                %%esplicitamente e inoltre $p=3$ ma se
                                %%non numeri bene i vertici non puoi
                                %%assumere che la $C$ sia fatta cosi'} 
is reduced to multi-output observability canonical form $(A_o, H_o)$ by resorting to a suitable transformation matrix $T$ (see
e.g. \cite{antbook}). 
(2) The matrix  $L_o$ that makes $A_o + L_o H_o$ nilpotent with  minimum possible nilpotency index is trivially 
obtained by imposing that  $A_o + L_o H_o$ is block-diagonal, with diagonal blocks that are 
(single-output) observability canonical forms with zero coefficients in the last column.
(3)  The desired observer gain matrix is then $L=T^{-1}L_o$.\\
It is worth noticing that there is a  transient phase, due to the presence of an estimation error, consisting of $3$ time steps, after which
the residual becomes zero 
(i.e.,  $\tau_0=3$).

\begin{figure}\label{fgiTAC}
\begin{center}
\hspace{-1cm}
\includegraphics[scale=0.3, angle= 90]{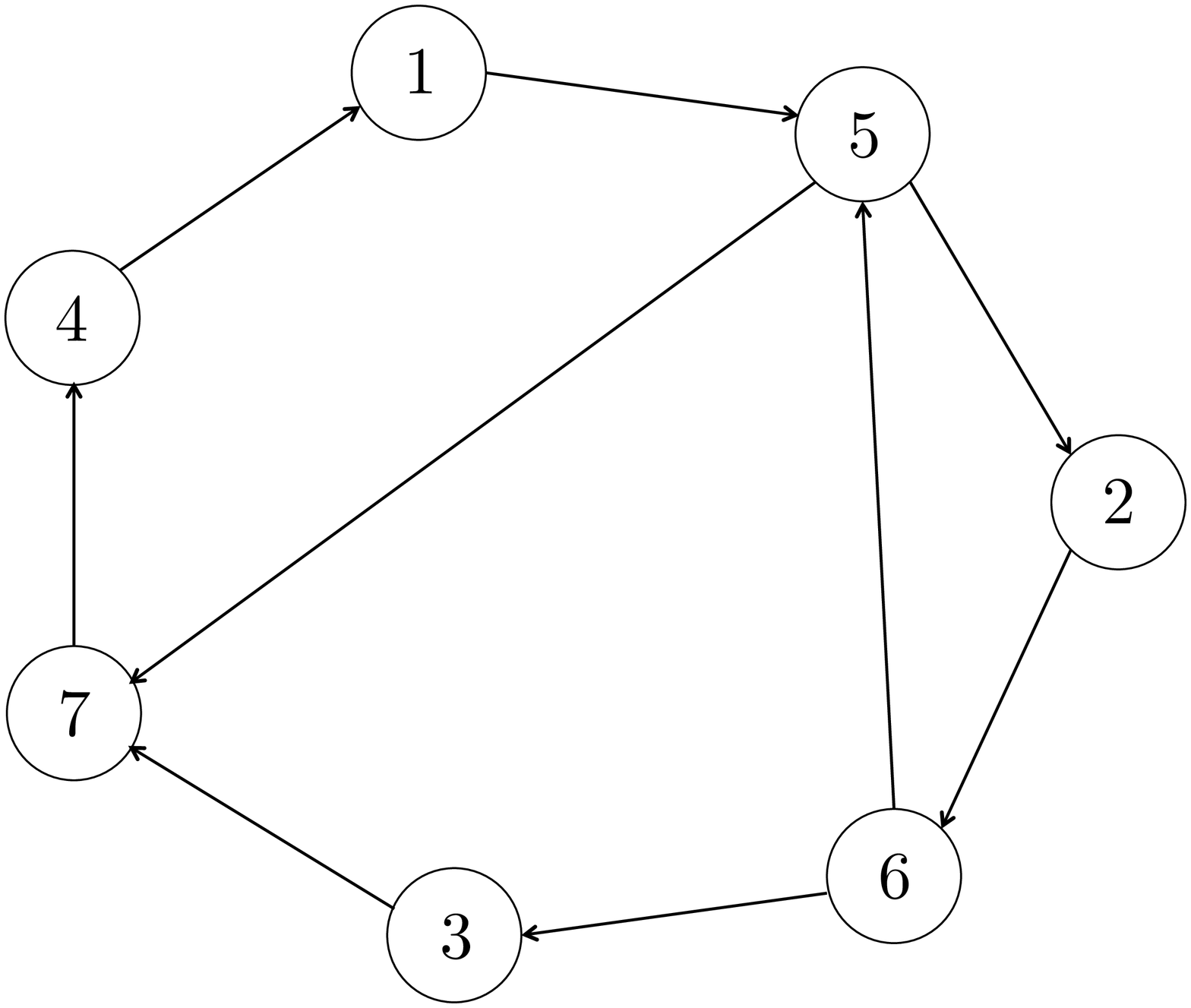}%
%\vspace{-2cm}  
\caption{A sketch of the graph of the illustrative example of Section \ref{example}.}
\end{center}
\end{figure}
%\bigskip

Two simulation results are plotted in
Figure 3.  The three curves 
represent, in both cases,  the multi-agent system outputs, namely the states of the first three agents.
On the other hand, instead of reporting the values of the residual ${\bf r}(t)$, we have chosen to 
represent  
with  black circles the 
 \emph{detection signal}   $d(t)$
 which is unitary if ${\bf r}(t)\ne0$ and $0$ if ${\bf r}(t)=0$. 
 %{\color{red} pero' stai sicuro che qualcuno brontolera'  se non stampa a colori..}
%$d(t)=sg({\bf r}(t)^\top {\bf r}(t))$ where the function $sg(\zeta)$ is
%equal to 0 if $\zeta=0$ and  $sg(\zeta)=1$ if $sg(\zeta)\neq 0$
%(namely it evaluates if ${\bf r(t)}$ is the zero vector).
%%
%The plot of $d(t)$ allows us to easily see when a fault has been
%detected, namely when it jumps from $0$ to $1$ for $t>\tau_0$.
%
%We show the simulation of two possible cases which refer to the case
%when the conditions of Prop. are not met.
In both simulations we have disconnected the edge $(6,5)$ in the interval
$[10, 14]$, and the edge  $(5,7)$ in the interval
$[20, 24]$.   It is worth noticing that $A$ and
 $\bar{A}_{7,5}$ do not have common eigenvalues (apart from $1$), so the original and the faulty networks are
  discernible, while $A$ and
  $\bar{A}_{5,6}$ have $\lambda=0.5$ and the corresponding eigenvector in
  common, so this case fails to satisfy condition (v) of
  Proposition \ref{classica} and, in turn, condition  (iv) of
  Proposition \ref{ostrega} and condition (i) of Proposition \ref{faultdetp_works}.

In the first simulation, corresponding to the upper plot of Figure  3, it is assumed
${\bf x}(0)=[ 10 \ -1 \ 1 \ 8 \ 5 \ 5 \ 12 ]^\top$ and $\hat {\bf
  x}(0)=0.$
After an initial transient of 3 steps, the estimated   state $\hat {\bf x}(3)$  is equal to the real state ${\bf x}(3)$
(while ${\bf x}(t)\ne \hat {\bf x}(t)$  for $t<3$, and this is the reason why the estimation signal is nonzero), and the detection signal is zero up to $t=11$ when the
disconnection of the edge $(6,5)$ is detected.
Then the link is restored, and 
the  disconnection of the edge $(5,7)$ is detected at time $t=22$.
%; it is worth
%remarking that the delay in the detection of a link failure is an
%intrinsic peculiarity of the detection problem using a subset
%of whole state vector. 
%In this simulation, the edge failure $(5,3)$ is
%first detected by the value of node $4$ (which is one hop from node
%$3$), while  the failure of edge $(3,7)$ is first detected by the
%value of node $2$  (which is two hops from node
%$7$). 

The  second simulation  shows what it may happen if the 
conditions of  Propositions \ref{ostrega} and  Proposition \ref{faultdetp_works}  are not met.  
It is assumed ${\bf x}(0)=[ -5 \ 5 \ 5 \ -5 \ -5 \ 5 \ -5 ]^\top$ and $\hat {\bf
  x}(0)=0$. After an initial transient phase, $\hat {\bf x}(3)={\bf x}(3)$,
however, the detection signal is zero up to time $t=22$, when the
second edge disconnection is detected, and this shows that the first
link disconnection remains undetected because of the special
structure of the graph topology and  the specific value of the system state at the time of the disconnection.

 % {\color{verde} Indeed, if the
 %  conditions of Propositions \ref{classica},  \ref{ostrega} and
 %  \ref{faultdetp_works} are not satisfied, extracting the
 %  information on the failed edge from data can be possible or not
 %  depending only on the running initial conditions.}
 %  {\color{red} cosa vuol dire? ma non dipende solo dalla struttura in generale; qui abbiamo tutti archi di peso unitario, ma 
 %  il problema emerge per certi pesi, certe disconnessioni e per certe condizioni $x(\tau)$}

\begin{figure}\label{SimRes}
\begin{center}
%\hspace{-1cm}
\includegraphics[scale=0.42]{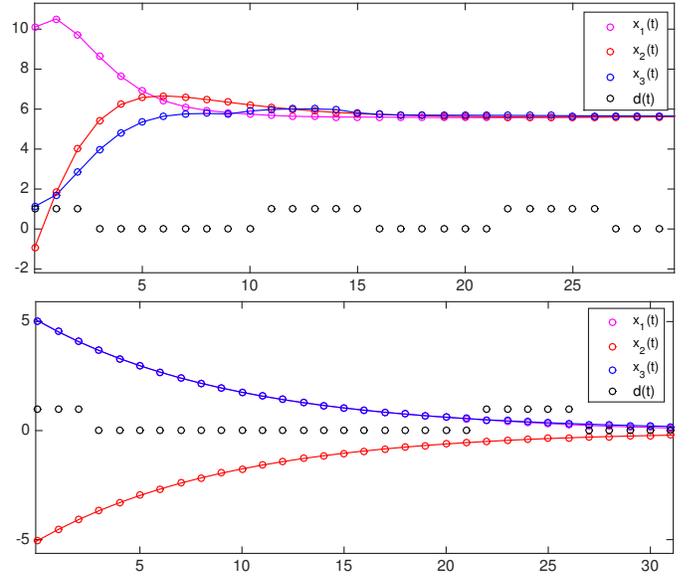}
%{FigDEFMagFin2.eps}%
%\vspace{-2cm}  
\caption{Simulation Results: output evolution and detection signal of
  system \eqref{Asist}
  under different initial conditions. In the bottom plot, the first
  link disconnection cannot be detected from system evolution.}
\end{center}
\end{figure}
%\bigskip
}
%  {\color{verde} Non metto dettagli sui colori perche' ho messo la legenda.}{\color{red} veramente la legenda non c'e'}{\color{red}
%   ma nella seconda quella purple e' sovrapposta alla blue??}
% {\color{verde} Si. Le c.i. per non far vedere il guasto hannop qieste
%   caratteristiche. Se zoommi, puoi vedere che dopo il secondo guasto
%   si separano (perchè il secondo guasto 'scombina' questo tipo di evoluzione)} 
  
% {\color{red}  Inoltre 
% bisogna spiegare bene cosa sono le curve e io personalmente eviterei
% di usare stars e diamonds. Infine, se si aggiornano i numeri dei nodi
% si devono eventualmente aggiornare anche gli archi che vengono
% disconnessi.}  ({\color{verde} Gia'!! Fatto.}) {\color{red} veramente io vedo ancora stars e diamonds}

\section{Conclusions}

{\color{black} In this paper we have addressed the problem of detecting and identifying an edge disconnection
in a discrete-time   consensus network, by assuming that the link failure does not compromise the strong 
connectedness of the underlying directed communication network.
The cases when the states of all the agents are available and when only a proper subset of them is available
are both considered, and sufficient conditions ensuring that the problem is solvable are provided. An example concludes
the  paper, illustrating both the case when detection from the measurement of the states of 3 of the 7 agents is possible and the case when it is not.

It is worth noticing that we have solved the discernibility problem from the first $p$ states by resorting to a full order dead-beat observer, but due to the structure of the state to output matrix the use of a reduced-order dead-beat observer would be  straightforward, and it would ensure the same performance in terms of nilpotency index.

Future research efforts will aim at finding an algorithm to efficiently identify the disconnected edge when only $p$ of the $N$ states are available, as it has been  done here in the case when all the agents are measured. Also, the case of noisy measurements and/or modelling errors   needs to be addressed.

 As mentioned  in the Introduction,  
 distributed fault detection and identification algorithms have been proposed by assuming that faults are additive.
 It would be of extreme interest to adapt such algorithms to the specific case when the fault 
is an edge disconnection, without losing the information about the specific nature and structure of the fault.
}
%***
%
%citare la differenza di passi necessari per fare detection e per fare identification
%
%menzionare il fatto che si potrebber utilizare uno stimatore dead beat di ordine ridotto per incrementare l;efficienza
%
%spiegare che l'ipotesi di avere fault da $\tau_0$ in poi serve per evitare di confondere un residuo non nullo dovuto al transitorio dello stimatore da un residuo non nullo dovuto a un fault
%
%  per avere detectability e' necessario che la matrice $A$ abbia autovalori tutti di molteplicita' geometrica unitaria (un miniblocco per ogni autovalore). Cio' garantisce che in linea di massima si potrebbe determinare una matrice $H\in {\mathbb R}^{1\times N}$ tale che $(A,H)$ sia osservabile, tuttavia in genere non possiamo limitarci a vettori canonici. Se poi dobbiamo garantire che la coppia $(\begin{bmatrix} A & 0\cr 0 & A_{hr}\end{bmatrix}, H)$ sia osservabile per ogni $\lambda\ne 1$ dobbiamo garantire che $\sigma)A)\cap \sigma(A_{hr})=\{1\}$... ma questo gia' lo spaevamo
%  
  
\section*{Appendix: {\color{black}A technical lemma}}

\begin{lemma}\label{tech4}
Consider  the positive irreducible  matrix $A=I_N - \kappa {\mathcal L}\in {\mathbb R}^{N\times N}$ and let
$J_A$ be its (real) Jordan form. Let $T\in {\mathbb R}^{N\times N}$ be the nonsingular transformation matrix  
such that
$$T^{-1} A T = J_A = \begin{bmatrix} 1& 0\cr 0 & \tilde J_A\end{bmatrix} = 
\begin{bmatrix} 1& 0 & \dots & 0 \cr 
0 & J_2   & \dots & 0\cr
\vdots & & \ddots &\vdots\cr 0 & 0 & \dots  & J_n
\end{bmatrix}, $$
where $J_i$ is a Jordan block corresponding to $\lambda_i$ and $\lambda_i\ne 1$ for every $i\in [2,n]$.
 Define $W$ as in \eqref{defW}, namely as
$$W := \begin{bmatrix} {\bf 0}_{N-1} & I_{N-1}\end{bmatrix} T^{-1}.$$
Then for every pair of distinct nodes $h, i \in {\mathcal V}=[1,N], i\ne h,$
$$\langle W {\bf e}_i \rangle \ne \langle  W {\bf e}_h \rangle.$$
\end{lemma}

\begin{proof}
Suppose, by contradiction, that there exist $h, i \in {\mathcal
  V}=[1,N], i\ne h,$  and   nonzero  $\alpha, \beta\in {\mathbb R}$, such that 
$$0= W [\alpha {\bf e}_i + \beta {\bf e}_h]= 
\begin{bmatrix} {\bf 0}_{N-1} & I_{N-1}\end{bmatrix} T^{-1}[\alpha {\bf e}_i + \beta {\bf e}_h].$$
Since $T$ is a nonsingular matrix,
there exists a vector ${\bf c} \ne 0$ such that
$\alpha {\bf e}_i + \beta {\bf e}_h = T {\bf c}.$ By replacing this expression in the previous identity we obtain
$$0= 
\begin{bmatrix} {\bf 0}_{N-1} & I_{N-1}\end{bmatrix} T^{-1}T {\bf c}= \begin{bmatrix} {\bf 0}_{N-1} & I_{N-1}\end{bmatrix} {\bf c}.$$
This amounts to saying that ${\bf c}= \gamma {\bf e}_1$ for some $\gamma\ne 0$ and
$$\alpha {\bf e}_i + \beta {\bf e}_h = T \gamma {\bf e}_1 = \gamma {\bf 1}_N.$$
 But this is clearly not possible, since the vector on the left hand side has two nonzero entries, while the one 
on the right hand side has all nonzero (and identical) entries.
\end{proof}
\smallskip


\begin{thebibliography}{10}

\bibitem{antsaklis2007special}
P.~Antsaklis and J.~Baillieul.
\newblock Special issue on technology of networked control systems.
\newblock {\em Proc. IEEE}, 95, no. 1:5--8, 2007.

\bibitem{antbook}
P.~J. Antsaklis and A.~N. Michel.
\newblock {\em Linear Systems}.
\newblock Birkh\"auser, Boston, MA, 2006.

\bibitem{baillieul2007control}
J.~Baillieul and P.J. Antsaklis.
\newblock Control and communication challenges in networked real-time systems.
\newblock {\em Proc. IEEE}, 95(1):9--28, 2007.

\bibitem{BattistelliTesi}
G.~Battistelli and P.~Tesi.
\newblock Detecting topology variations in networks of linear dynamical
  systems.
\newblock {\em IEEE Trans. Control Network Systems}, 5, no. 3:1287--1299, 2018.

\bibitem{Berman-Plemmons}
A.~Berman and R.J. Plemmons.
\newblock {\em Nonnegative matrices in the mathematical sciences}.
\newblock Academic Press, New York, 1979.

\bibitem{DistCtrlRobotNetw}
F.~Bullo, J.~Cort\'es, and S.~Mart{\'\i}nez.
\newblock {\em Distributed Control of Robotic Networks}.
\newblock Princeton University Press, 2009.

\bibitem{cardoso2007laplacian}
D.~M. Cardoso, C.~Delorme, and P.~Rama.
\newblock Laplacian eigenvectors and eigenvalues and almost equitable
  partitions.
\newblock {\em European J. Combinatorics}, 28(3):665--673, 2007.

\bibitem{chong2003sensor}
C.-Y. Chong and S.P. Kumar.
\newblock Sensor networks: evolution, opportunities, and challenges.
\newblock {\em Proc. IEEE}, 91(8):1247--1256, 2003.

\bibitem{costanzo2017using}
J.~Costanzo, D.~Materassi, and B.~Sinopoli.
\newblock Using Viterbi and Kalman to detect topological changes in dynamic
  networks.
\newblock In {\em Proc. 2017 American Control Conference}, pages 5410--5415.
  IEEE, 2017.

\bibitem{Cullen}
C.G. Cullen.
\newblock {\em Matrices and Linear Transformations}.
\newblock Dover Publications; 2nd edition, 1990.

\bibitem{davoodi2014distributed}
M.R. Davoodi, K.~Khorasani, H.A. Talebi, and H.R. Momeni.
\newblock Distributed fault detection and isolation filter design for a network
  of heterogeneous multiagent systems.
\newblock {\em IEEE Trans. Control Systems Technology}, 22(3):1061--1069, 2014.

\bibitem{dhal2013link}
R.~Dhal, J.~A. Torres, and S.~Roy.
\newblock Link-failure detection in network synchronization processes.
\newblock In {\em Proc. 2013 IEEE Global Conference on Signal and Information
  Processing}, pages 779--782,2013.

\bibitem{dhal2015detecting}
R.~Dhal, J.~Abad Torres, and S.~Roy.
\newblock Detecting link failures in complex network processes using remote
  monitoring.
\newblock {\em Physica A: Statistical Mechanics and its Applications},
  437:36--54, 2015.

\bibitem{egerstedt2012interacting}
M.~Egerstedt, S.~Martini, M.~Cao, K.~Camlibel, and A.~Bicchi.
\newblock Interacting with networks: How does structure relate to
  controllability in single-leader, consensus networks?
\newblock {\em IEEE Control Systems Magazine}, 32(4):66--73, 2012.

\bibitem{BookFarina}
L.~Farina and S.~Rinaldi.
\newblock {\em Positive linear systems: theory and applications}.
\newblock Wiley-Interscience, Series on Pure and Applied Mathematics, New York,
  2000.

\bibitem{Fiedler_consensus}
M.~Fiedler.
\newblock Algebraic connectivity of graphs.
\newblock {\em Czechoslovak Mathematical J.}, 23:298--305, 1973.

\bibitem{gharavi2003special}
H.~Gharavi and S.P. Kumar.
\newblock Special issue on sensor networks and applications.
\newblock {\em Proc. IEEE}, 91(8):1151--1153, 2003.

\bibitem{HornJohnson}
R.A. Horn and C.R. Johnson.
\newblock {\em Matrix Analysis}.
\newblock Cambridge Univ. Press, Cambridge (GB), 1985.

\bibitem{leonard2007collective}
N.E. Leonard, D.A. Paley, F.~Lekien, R.~Sepulchre, D.M. Fratantoni, and R.E.
  Davis.
\newblock Collective motion, sensor networks, and ocean sampling.
\newblock {\em Proc. IEEE}, 95(1):48--74, 2007.

\bibitem{li2017robust}
Z.~Li and J.~Chen.
\newblock Robust consensus of linear feedback protocols over uncertain network
  graphs.
\newblock {\em IEEE Trans. Automatic Control}, 62(8):4251--4258, 2017.

\bibitem{miah2014nonuniform}
S.~Miah, B.~Nguyen, A.~Bourque, and D.~Spinello.
\newblock Nonuniform coverage control with stochastic intermittent
  communication.
\newblock {\em IEEE Trans. Automatic Control}, 60(7):1981--1986, 2014.

\bibitem{OlfatiFaxMurray}
R.~Olfati-Saber, J.A. Fax, and R.M. Murray.
\newblock Consensus and cooperation in networked multi-agent systems.
\newblock {\em Proc.  IEEE}, 95, no. 1:215--233, 2007.

\bibitem{OF-Murray2004}
R.~Olfati-Saber and R.M. Murray.
\newblock Consensus problems in networks of agents with switching topology and
  time-delays.
\newblock {\em IEEE Trans. Automatic Control}, 49, no. 9:1520 --1533, 2004.

\bibitem{OReilly}
J.~O'Reilly.
\newblock {\em Observers for Linear Systems}.
\newblock Academic Press, 1983.

\bibitem{pandey2019diffusion}
P.~K. Pandey, B.~Adhikari, and S.~Chakraborty.
\newblock A diffusion protocol for detection of link failure and utilization of
  resources in multi-agent systems.
\newblock {\em IEEE Trans. Network Science and Engineering}, 2019.

\bibitem{pasqualetticarliIFAC}
F.~Pasqualetti, R.~Carli, A.~Bicchi, and F.~Bullo.
\newblock Distributed estimation and detection under local information.
\newblock In {\em Proc. 2nd IFAC Workshop on Distributed Estimation and Control
  in Networked Systems}, pages 263--268, Annecy, France, 2010.

\bibitem{pasqualettiTAC2013}
F.~Pasqualetti, F.~Dorfler, and F.~Bullo.
\newblock Attack detection and identification in cyber-physical systems.
\newblock {\em IEEE Trans. Automatic Control}, 58, no. 11:2715--2729, 2013.

\bibitem{pasqualetti2015control}
F.~Pasqualetti, F.~Dorfler, and F.~Bullo.
\newblock Control-theoretic methods for cyberphysical security: Geometric
  principles for optimal cross-layer resilient control systems.
\newblock {\em IEEE Control Systems Magazine}, 35, no. 1:110--127, 2015.

\bibitem{patil2019indiscernible}
D.~Patil, P.~Tesi, and S.~Trenn.
\newblock Indiscernible topological variations in dae networks.
\newblock {\em Automatica}, 101:280--289, 2019.

\bibitem{qin2017recent}
J.~Qin, Q.~Ma, Y.~Shi, and L.~Wang.
\newblock Recent advances in consensus of multi-agent systems: A brief survey.
\newblock {\em IEEE Trans. Industr. Electronics}, 64(6):4972--4983, 2017.

\bibitem{RahimianCDC2012}
M.~Rahimian, A.~Ajorlou, and A.~Aghdam.
\newblock Detectability of multiple link failures in multi-agent systems under
  the agreement protocol.
\newblock In {\em Proc. 2012 IEEE Conf. Decision Control}, pages 118--123,
  Maui, HI, USA, 2012.

\bibitem{rahimian2015failure}
M.~A. Rahimian and V.M. Preciado.
\newblock Failure detection and isolation in integrator networks.
\newblock In {\em Proc. 2015 American Control Conference}, pages 677--682,
2015.

\bibitem{SurveyConsensus2005}
W.~Ren, R.~W. Beard, and E.~M. Atkins.
\newblock A survey of consensus problems in multi-agent coordination, 2005.
\newblock In {\em Proc. 2005 American Control Conference}, pages 1859--1864,
2005.

\bibitem{ren2007distributed}
W.~Ren and R.W. Beard.
\newblock {\em Distributed Consensus in Multi-vehicle Cooperative Control:
  Theory And Applications}.
\newblock Springer, 2007.

\bibitem{RenBeardAtkins}
W.~Ren, R.W. Beard, and E.M. Atkins.
\newblock Information consensus in multivehicle cooperative control.
\newblock {\em IEEE Control Systems Magazine}, 27 (2):71--82, 2007.

\bibitem{ren2010distributed}
W.~Ren and Y.~Cao.
\newblock {\em Distributed coordination of multi-agent networks: emergent
  problems, models, and issues}.
\newblock Springer Science \& Business Media, 2010.

\bibitem{SontagBook}
E.D. Sontag.
\newblock {\em Mathematical Control Theory. Deterministic Finite Dimensional
  Systems.}
\newblock Springer-Verlag, New York, 2nd edition, 1998.

\bibitem{teixeira2014distributed}
A.~Teixeira, I.~Shames, H.~Sandberg, and K.H. Johansson.
\newblock Distributed fault detection and isolation resilient to network model
  uncertainties.
\newblock {\em IEEE Trans. Cybernetics}, 44(11):2024--2037, 2014.

\bibitem{torres2015detecting}
J.~A Torres, R.~Dhal, and S.~Roy.
\newblock Detecting link failures in complex network processes using remote
  monitoring.
\newblock In {\em Proc. 2015 American Control Conference}, pages 189--194, 2015.

\bibitem{ECC2019}
M.E. Valcher.
\newblock Consensus in the presence of communication faults.
\newblock In {\em Proc. 2019 European Control Conf.}, pages 1062--1067, Napoli,
  Italy, 2019.

\bibitem{ICSTCC2019}
M.E. Valcher and G.~Parlangeli.
\newblock On the effects of communication failures in a multi-agent consensus
  network.
\newblock In {\em Proc. 23rd International Conference on System Theory,
  Control and Computing}, pages 709--720, Sinaia, Romania, 2019.

\bibitem{wah2007synchronization}
W.C. Wah.
\newblock {\em Synchronization in complex networks of nonlinear dynamical
  systems}.
\newblock World Scientific, 2007.

\bibitem{Xue}
M.~Xue.
\newblock Designing local inputs to identify link failures in a diffusive
  network: A graph perspective.
\newblock In {\em Proc. 58th Conference on Decision and Control}, pages
  5525--55230, Nice, France, 2019.

\bibitem{you2013consensus}
K.~You, Z.~Li, and L.~Xie.
\newblock Consensus condition for linear multi-agent systems over randomly
  switching topologies.
\newblock {\em Automatica}, 49(10):3125--3132, 2013.

\bibitem{zhao2004wireless}
F.~Zhao and L.~Guibas.
\newblock {\em Wireless sensor networks: {A}n information processing approach}.
\newblock Morgan Kaufmann, 2004.

\end{thebibliography}
\end{document}